\documentclass[journal]{IEEEtran}
\usepackage{pifont}
\usepackage{diagbox}
\usepackage{}
\usepackage{mathrsfs}
\usepackage[noadjust]{cite}
\usepackage{graphicx}
\usepackage{float}
\usepackage{subfigure}
\usepackage{amsthm}
\usepackage{amsmath}
\usepackage{bm}
\usepackage{enumerate}
\usepackage{amssymb}
\usepackage{fixltx2e}
\usepackage{color}
\usepackage{hhline}
\usepackage{algpseudocode}
\usepackage{amsmath}
\usepackage{graphics}
\usepackage{epsfig}
\usepackage{setspace}
\usepackage{multirow}
\usepackage{ragged2e}
\usepackage{makecell}
\usepackage[table,xcdraw]{xcolor}

\newtheorem{corollary}{Corollary}
\newtheorem{remark}{Remark}
\usepackage{bbding}
\usepackage[ruled]{algorithm2e}
\usepackage{array}
\usepackage{booktabs}
\usepackage{graphics}
\usepackage{epsfig}
\usepackage{makecell}

\allowdisplaybreaks[4]
\newcounter{MYtempeqncnt}  
\newtheorem{theorem}{Theorem}

\begin{document}

\title{Physical Layer Security in AmBC-NOMA Networks with Random Eavesdroppers}
\author{Xinyue~Pei, Xingwei~Wang, Min~Huang,  Yingyang~Chen,~\IEEEmembership{Senior Member,~IEEE},\\
	Xiaofan~Li,~\IEEEmembership{Senior Member,~IEEE}, and Theodoros A. Tsiftsis,~\IEEEmembership{Senior Member,~IEEE}
	
	\thanks{X. Pei and X. Wang are with School of Computer Science and Engineering, Northeastern University, Shenyang 110819, China (e-mail:
		peixy@cse.neu.edu.cn, wangxw@mail.neu.edu.cn).}
		\thanks{M. Huang is with the College of Information Science
			and Engineering, Northeastern University, Shenyang 110819, China (e-mail:
			mhuang@mail.neu.edu.cn).}
	\thanks{Y. Chen is with the 
		College of Information Science and Technology, Jinan University, Guangzhou 510632, China (e-mail:
		chenyy@jnu.edu.cn). }
	\thanks{X. Li is with the School of Intelligent Systems Science and
		Engineering, Jinan University, Zhuhai 519070, China (e-mail: lixiaofan@jnu.edu.cn).}
	\thanks{T. A. Tsiftsis is with Department of Informatics Telecommunications, University of Thessaly, 
		Lamia 35100, Greece, and also with the Department of Electrical and Electronic Engineering, 
		University of Nottingham Ningbo China, Ningbo 315100, China (e-mail: tsiftsis@uth.gr).}
}
\markboth{SUBMITTED TO IEEE INTERNET OF THINGS JOURNAL}{}
\maketitle

\begin{abstract}	
In this work, we investigate the physical layer security (PLS) of ambient backscatter communication non-orthogonal multiple access (AmBC-NOMA)   networks where non-colluding eavesdroppers (Eves) are randomly distributed. In the proposed system, a base station (BS) transmits a superimposed signal to a typical
NOMA user pair, while a backscatter device~(BD)  simultaneously transmits its unique signal by reflecting and modulating the BS's signal. Meanwhile, Eves passively attempt to wiretap the ongoing transmissions.  Notably, the number and locations of Eves are unknown, posing a substantial security threat to the system. To address this challenge, the BS injects artificial noise (AN) to mislead the Eves, and a protected zone is employed to create an Eve-exclusion area around the BS. 
Theoretical expressions for outage probability (OP) and intercept probability (IP)  are provided to evaluate the system's reliability-security trade-off. Asymptotic behavior at high signal-to-noise ratio (SNR) is further explored, including the derivation of diversity orders for the OP. Numerical results validate the analytical findings  through extensive simulations, demonstrating that both the AN injection and protected zone can effectively enhance PLS. Furthermore, analysis and insights of different key parameters, including transmit SNR, reflection efficiency at the BD, power allocation coefficient, power fraction allocated to desired signal, Eve-exclusion area radius, Eve distribution density, and backscattered AN cancellation efficiency, on OP and IP are also provided.

\end{abstract}

\maketitle
\section{Introduction}
Internet of Things (IoT) has become a core element of sixth-generation (6G) mobile networks \cite{nguyen20216g}.  However, the explosive growth of IoT devices presents significant challenges, including the need for massive connectivity and management of spectrum and energy constraints \cite{cisco2020cisco}. 
In this regard, traditional orthogonal multiple access (OMA) systems, which assign unique time, frequency, or code resources to each user, face the risk of overload.  Non-orthogonal multiple access (NOMA), on the other hand, uses techniques like superposition coding at the transmitter and successive interference cancellation (SIC) at the receiver to improve spectral efficiency (SE) and support massive connectivity \cite{pei2022next}. Therefore, integrating NOMA with IoT is emerging as a promising direction for future 6G networks~\cite{Xin2021iot}.

Beyond enhancing SE, maximizing energy efficiency (EE) is another crucial challenge in IoT, particularly in green IoT applications \cite{arshad2017green}. Nowadays, ambient backscatter communication (AmBC) has become a promising solution for green IoT, thanks to its low energy consumption and ability to avoid occupying  dedicated spectrum \cite{liu2013ambient, zhang2018spectrum}. In AmBC systems, a backscatter device (BD) reflects signals from an ambient radio frequency (RF) source to the receiver, while superimposing its own signal onto the RF signals in the environment (such as digital television broadcasting or cellular network signals) without requiring any oscillatory circuitry. In this manner, AmBC enables spectrum- and energy-efficient communication~\cite{liu2013ambient}.
Clearly, combining NOMA with AmBC can utilize their respective advantages, enabling the realization of battery-free IoT networks.  Specifically, in AmBC-NOMA, the BD's additional signal is decoded only at the closer user via SIC, while acting as interference for the farther user \cite{zhang2019backscatter}.

Although AmBC-NOMA networks present notable advantages, their broadcast nature exposes them to eavesdropping and malicious attacks \cite{li2021physical}. Traditional key-based encryption, however, may be impractical due to challenges such as authentication latency, energy consumption, and key management, particularly in high-density environments \cite{wang2019physical}. 
To address these risks, physical layer security (PLS), originally proposed by Wyner in 1975,  is considered a promising solution for enhancing communication security~\cite{wyner1975wire}. Unlike key-based encryption, PLS exploits the inherent randomness of wireless channels to safeguard data at the physical layer, avoiding complex cryptographic operations \cite{dong2009improving,chen2016secure}.

Specifically, Wyner’s pioneering work established secrecy capacity, identified as the gap between the main and wiretap channel capacities, as a critical metric in PLS. A positive secrecy capacity ensures secure communication, while a negative one indicates a high risk of successful eavesdropping  \cite{leung1978gaussian}. To quantify this risk, intercept probability (IP) has emerged as a core metric for security evaluation \cite{zou2015improving}. It is worth noting that a trade-off between reliability and security is common \cite{zou2013optimal}, as increasing transmission power to reduce the main link’s outage probability (OP) may unintentionally raise the risk of eavesdropping.  Consequently, balancing IP and OP has been extensively studied across various networks, including AmBC-NOMA systems \cite{li2021hardware}.

Various techniques, including the injection of artificial noise~(AN) \cite{negi2005secret} and cooperative jamming (CJ) \cite{tekin2008general}, protect legitimate communication from eavesdropping to further enhance PLS.
These techniques aim to degrade the reception quality of potential eavesdroppers (Eves) without affecting legitimate users. Among them, AN is particularly effective. By deliberately adding noise to the transmitted signal, AN can disrupt the decoding capabilities of Eves without requiring additional hardware \cite{goel2008guaranteeing}. This property makes AN particularly suitable for low-complexity, low-power AmBC applications~\cite{li2021hardware,Li2022secureV2P}. 

\subsection{Related Works}
In this subsection, we review the related works focusing on PLS design in both AmBC and AmBC-NOMA systems.
\subsubsection{PLS design in AmBC network}
Given the significance of studying PLS in AmBC networks, this research topic sparked widespread recent interest. Specifically, in \cite{li2021physical}, the authors proposed a novel cognitive radio (CR)-based AmBC framework and examined the considered system's OP and IP. In \cite{muratkar2021physical}, a practical AmBC network with mobile nodes was examined, focusing on the secrecy performance under imperfect channel estimation. Then in \cite{wang2022physical}, the asymptotic outage and intercept behaviors for two-way AmBC networks were comprehensively analyzed.
Furthermore, in \cite{li2023cooperativephysical}, theoretical expressions for secrecy energy efficiency (SEE) as well as secrecy outage probability (SOP)  were derived in a cooperative AmBC network utilizing a decode-and-forward relay.
In addition,~\cite{chu2024countering} proposed a novel message-splitting technique to establish an anti-eavesdropping framework for AmBC networks. 
In scenarios with multiple BDs, \cite{liu2021secrecy} introduced an approach for selecting the optimal tag to improve secrecy performance and provided a closed-form expression for SOP.  Building upon this, \cite{liu2022secrecy} analyzed ergodic secrecy capacity~(ESC) along with SOP in self-powered networks with multiple tags, considering Nakagami-$m$ fading channels.
Regarding intelligent transportation systems, \cite{Li2022secureV2P} addressed the PLS of AmBC-assisted vehicle-to-pedestrian networks through utilizing AN injection. In \cite{jia2023secrecy}, the authors designed secure communications in an AmBC-based intelligent transportation system in conjunction with a passive Eve and a multiple-antenna RF source, employing CJ. Finally, \cite{jia2024secrecy} investigated the SOP of unmanned aerial vehicle-assisted AmBC systems employing CJ.


\subsubsection{PLS design in AmBC-NOMA} 
The authors in \cite{khan2020secure}  proposed an optimization framework to improve communication security in multi-cell AmBC-NOMA networks. In the referenced paper 	 \cite{li2020secrecy}, the influence of in-phase and quadrature imbalance (IQI) on the secrecy performance of AmBC-NOMA networks was studied. Based on this, \cite{li2021hardware}  investigated AmBC-NOMA networks under practical conditions, i.e., residual hardware impairments, imperfect SIC and channel state information (CSI), and evaluated their combined impact on system reliability and security.
	Another work \cite{li2021cognitive} focused on the reliability and security of CR-assisted AmBC-NOMA networks under IQI, specifically for Internet of Vehicle-based maritime transportation systems. Subsequently, \cite{zhao2024performance} explored how tag sensitivity and IQI influence the reliability and security of cognitive AmBC-NOMA systems. Furthermore, the reliability and security of AmBC-NOMA systems under IQI assumptions and Nakagami-$m$ fading were comprehensively studied in \cite{sun2024secrecy}.
	Finally, \cite{chrysologou2024reliability} investigated the  inherent trade-off between reliability and security for uplink AmBC-NOMA systems.


\begin{table}[!t]
	\begin{Center}					
		\caption{Comparison between our work with the state-of-the-art.}  
		\label{Literature_Review}
		\scalebox{0.87}{
			\begin{tabular}{|c|c|c|c|c|}
				\hline\hline
				Ref                                                              & AmBC                                     & NOMA                  & \begin{tabular}[c]{@{}c@{}}Random\\ Eves\end{tabular}  & Metrics                                                                \\ \hline \hline
				\multicolumn{1}{|c|}{\cite{muratkar2021physical,liu2021secrecy,jia2023secrecy,jia2024secrecy}}                  & Yes                                                                      & No               & No                                                                                                               & SOP                                                                   \\ \hline
				\multicolumn{1}{|c|}{\cite{wang2022physical,li2021physical,Li2022secureV2P,chu2024countering}}                  & Yes                                                                      & No                & No                                                                                                               & OP IP                                                                                                                                     \\ \hline  
				\multicolumn{1}{|c|}{\cite{li2023cooperativephysical}}                  & Yes                                                                      & No                & No                                                                                                               & SOP SEE                                                                                                                                  \\ \hline
				\multicolumn{1}{|c|}{\cite{liu2022secrecy}}                  & Yes                                                                      & No                 & No                                                                                                               & SOP ESC                                                                  \\ \hline
				{\cite{khan2020secure}}                                                                                         & Yes                                                                      & Yes                  & No                                                                                                                & ESC                  \\ \hline 
				\multicolumn{1}{|c|}{{\cite{li2020secrecy,li2021cognitive,sun2024secrecy,li2021hardware,zhao2024performance,chrysologou2024reliability}}}                                                          &  Yes                                                                   & Yes                  & No                                                                                                             & OP IP                                                                                                                                                                                                                                                                                                                   \\ \hline
				Our work                                                                        & Yes                                                                      & Yes                & Yes                                                                                                             & OP IP                                                                   \\ \hline\hline
		\end{tabular}}
	\end{Center}
\end{table}	

\subsection{Motivations and Contributions}
\subsubsection{Motivations and Novelty}
It is worth noting that prior works have relied on predefined locations and numbers of Eves. However, such assumptions may not hold for passive Eves, which only overhear data transmissions without actively interfering. In such cases, identifying the specific positions and numbers of Eves becomes extremely challenging, especially when they are distant from the base station (BS). Therefore, the assumptions utilized in prior works may become unrealistic in passive eavesdropping scenarios~\cite{shiu2011physical}.

Without prior information about the Eves, we can treat
them as completely randomly distributed in the region \cite{pinto2011secure}. 
From this perspective,  stochastic geometry has proven to be an effective tool for modeling the random locations and numbers of unknown Eves \cite{chiu2013stochastic}.  Specifically, it can statistically model both small-scale and large-scale fading in wiretap channels~\cite{liu2017enhancing}. In particular, among all the various stochastic geometry methods used for modeling Eves, one widely adopted model is the homogeneous Poisson point process (HPPP)  \cite{pei2024secrecy}.\footnote{This arises from two key characteristics of the HPPP distribution \cite{pinto2011secure}: 1) it maximizes entropy among all homogeneous point processes, and 2) the probability density function (PDF) of a node's location is uniformly distributed within a given region.}
However, to the best of the authors' knowledge, no prior work on PLS in AmBC-NOMA networks has considered random Eves. As a result, the impact of key system parameters remains unclear within the stochastic geometry framework, arising the need to investigate the secrecy performance under this assumption.


\subsubsection{Contributions}
To address the  above research  gap, we explore the secure transmission of an AmBC-NOMA system in the presence of random Eves.
In the proposed system, it is assumed that the BS has the capability to detect and eliminate  the impact of nearby Eves by  using devices like metal detectors and X-ray scanners prior to transmission \cite{zhang2013enhancing,romero2013phy,liu2015ergodic,chae2014enhanced}.
Based on this assumption, a protected zone is established as a disc region around the BS, free of Eves. The key contributions of this paper are outlined below: 
\begin{itemize}
	\item 
	We analyze the secrecy performance of a novel secure AmBC-NOMA system  in the presence of HPPP-distributed Eves.
To further enhance PLS, a protected zone, namely, Eve-exclusion area is introduced. AN injection is also employed to improve security performance.

	\item We derive theoretical expressions for both the OP of  legitimate link and IP of  wiretap link. Based on the aforementioned expressions, we study
	both outage and interception behaviors to jointly evaluate the reliability and security of the proposed system.

\item To provide more insights, we also obtain the asymptotic expressions for both OP and IP at high signal-to-noise ratio (SNR), and attain the diversity order of the OP in this regime. The obtained results reveal the presence of error floors in the OP.
	\item  The accuracy of the derived results is validated through simulations. Additionally, numerical results explore the impact of various key parameters. The advantages of using AN and the protected zone are then demonstrated. Finally, the trade-off between reliability and security is observed.
\end{itemize}	
\subsection{Outline and Notations}
	The rest of this paper is organized as follows. Section~II provides an overview of the proposed system model. Section \uppercase\expandafter{\romannumeral3} provides the OP analysis.  Section \uppercase\expandafter{\romannumeral4} provides the IP analysis. Section \uppercase\expandafter{\romannumeral5} analyzes the numerical results. Finally, the conclusion of this work is presented in Section~\uppercase\expandafter{\romannumeral6}.  
	
\textit{Notations}: The operations $\Pr(\cdot)$, $|\cdot|$,  and $\mathbb{E}\{\cdot\}$ denote the probability, the absolute value,  and the expectation,
	respectively. $F_X(\cdot)$ and $f_X(\cdot)$ respectively denote the cumulative distribution function (CDF) and PDF of a random variable $X$.
	A complex Gaussian distribution $Y$ with zero mean and variance $\Omega$ is represented by $Y\sim\mathcal{CN}(0,\Omega)$, with $F_{|Y|^2}(x)=1-\exp(-\frac{x}{\Omega})$ and $f_{|Y|^2}(x)=\frac{1}{\Omega}\exp(-\frac{x}{\Omega})$.   $\Gamma(s,x)=\int_{x}^{\infty}t^{s-1}\exp(-t)dt$ is  the upper incomplete gamma function, $K_v(\cdot)$ is the modified bessel function of the second kind,
	 and $\text{Ei}(\cdot)$ is the exponential integral function  \cite{gradshteyn2014table}.
\section{System Model}
In this section, we first introduce  the proposed system model with unknown Eves,  then present the transmission model, and finally describe the signal decoding process along with the mathematical expression for the signal-to-interference plus noise ratio (SINR).
\subsection{System Overview}
In this work, we consider an AmBC-NOMA scenario as depicted in Fig. \ref{system_model}, which consists of a BS, a BD, a near user (denoted by $U_N$),  a far  user (denoted by $U_F$),\footnote{Due to the strong interference between users, it may be challenging to jointly apply NOMA to all users in practice. A low-complexity solution is to divide users into orthogonal pairs,  with NOMA applied within each pair. Consequently, in this work, we focus on a typical NOMA user pair to explore the trade-off between reliability and security.
	
} and multiple \textit{non-colluding passive} Eves. All the nodes are equipped with a single antenna.
The main components of the system are specified as follows:
\begin{figure}[!t]
	\centering
	\includegraphics[width=3.4in]{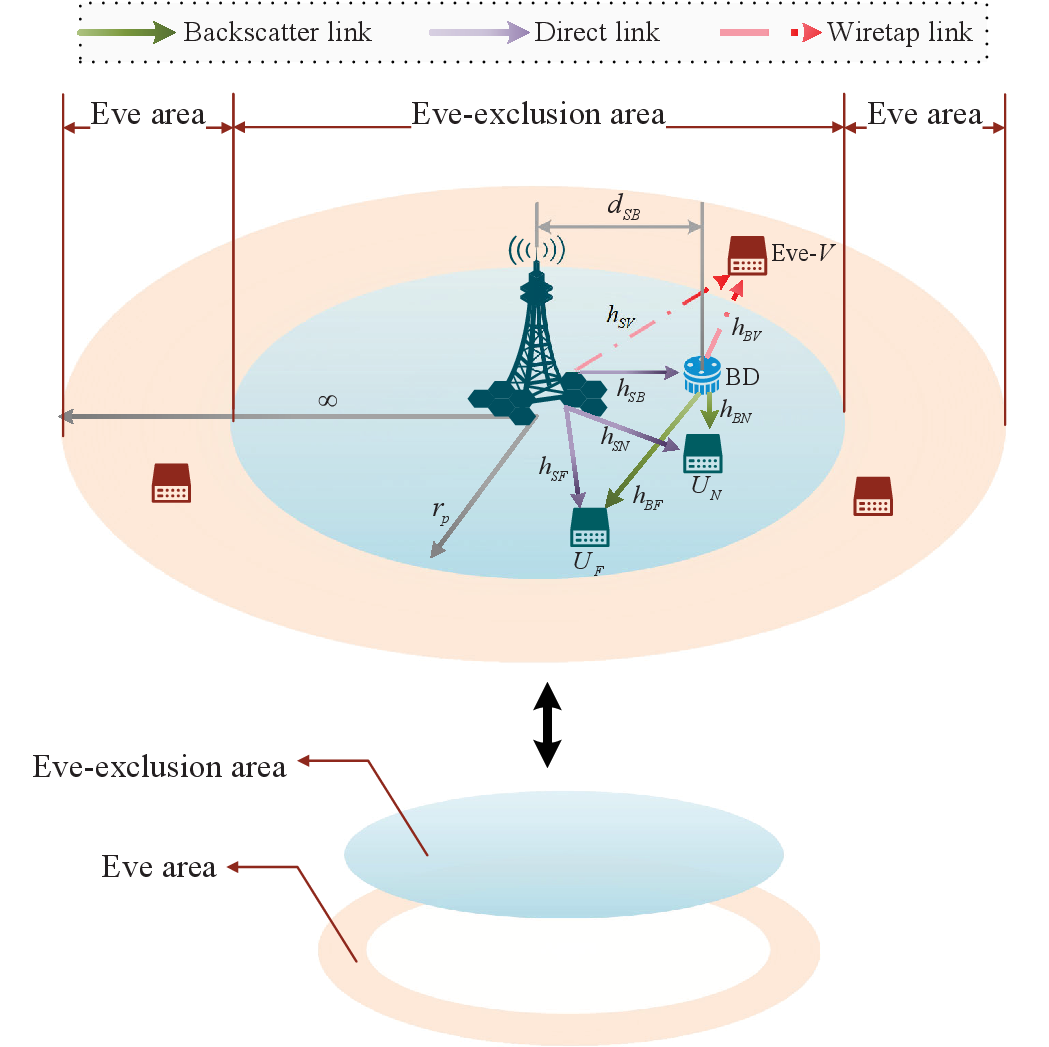}
	\caption{System model.} 	
	\label{system_model}
\end{figure}
 \begin{itemize}
\item \emph{BS and legitimate users:} 
In this work, we consider a single-cell scenario, where a BS communicates with legitimate users via the NOMA protocol. Without loss of generality, we assume that the BS  is positioned at the center, with legitimate users located  at fixed locations. Specifically, the BS employs AN injection to  prevent eavesdropping.
\item \emph{Eve-exclusion area:} We assume that the Eve-exclusion area is a disc region with radius $r_p$, entirely free of Eves. To further enhance PLS, legitimate users should be positioned within this region. 

 
	 	\item \emph{Eves:} While legitimate transceivers are communicating, multiple \textit{non-colluding passive} Eves are silently listening, aiming to degrade or even interrupt the ongoing communication. Recall that the Eves are unknown, we model their spatial distribution using an  HPPP, denoted by $\Phi_{\text{E}}$ with density $\lambda_e$. Then the distribution area for Eves (referred to as the Eve area) can be represented as an infinite two-dimensional plane, excluding the Eve-exclusion area. 
  	\item \emph{BD:}  To further enhance system performance, we place a BD at a fixed distance $d_{SB}$ from the BS. Given that the BD can be considered as a cooperative communication device, we assume that the BD is located near the BS. Additionally, since the BS can be equipped with high-performance detection equipment, we assume that $d_{SB}\ll r_p$. Consequently, the distance from the $V$-th Eve (denoted by $\text{Eve-}V$, $V\in\Phi_{\text{E}}$) to the BS (defined as $d_{SV}$) is approximately equal to the distance from $\text{Eve-}V$ to the BD (defined as $d_{BV}$), i.e., $d_{SV}\approx d_{BV}$.	
  		\footnote{This approximation can be derived using the cosine law. Given that \(d_{SV} \geq r_p \gg d_{SB}\), we have 
  			$	d_{BV} = \sqrt{d_{SV}^2 + d_{SB}^2 - 2d_{SV} d_{SB} \cos(\angle VSB)}$.  Then we can conclude that \(d_{BV} \approx d_{SV}\). This is a general assumption in cooperative networks and has been widely used in the literature \cite{liu2016cooperative,zhang2024noma}. }
  \end{itemize}

\subsection{Transmission Model}
It is assumed that all channels are subject to Rayleigh fading, and the channel coefficients from  BS to $U_N$, $U_F$, BD, and $\text{Eve-}V$ are respectively denoted by $h_{SN}$, $h_{SF}$, $h_{SB}$, and $h_{SV}$. Similarly, the channel coefficients from BD to $U_N$, $U_F$, and $\text{Eve-}V$ are respectively denoted by $h_{BN}$, $h_{BF}$, and $h_{BV}$. Accordingly,
	the channel coefficient $h_{\iota}$ ($\iota\in\{SN,SF,SB,SV,BN,BF,BV\}$)  can be expressed as
	$h_{\iota}=\sqrt{\lambda_{\iota}}g_{\iota}$,
	where $\lambda_{\iota}\triangleq d_{\iota}^{-\alpha}$ represents the large-scale fading; $\alpha$ is the path loss exponent; $d_{\iota}$ is the distance of the corresponding link;  $g_{\iota}\sim\mathcal{CN}(0,1)$ is a normalized complex Gaussian random variable.

To improve the PLS levels of the system, we herein inject an AN message $z$ with $\mathbb{E}\{|z|^2\}=1$ at the BS. Then, utilizing the NOMA protocol, the BS transmits a superimposed signal
\begin{align}\label{transmitted_signal}
s=\underbrace{\sqrt{\theta Pa_N}s_N +\sqrt{\theta Pa_F}s_F}_{\text{desired signal}} +\underbrace{\sqrt{(1-\theta)P} z}_{\text{AN signal}},
\end{align} to $U_N$ and $U_F$, where $P$ is the total transmit power; 
$\theta$ is the power fraction allocated to desired signal, which should satisfy $0.5 \leq \theta < 1$ to ensure the quality of desired signal; 
$s_N$ and $s_F$ respectively denote the intended signals for $U_N$ and $U_F$, with $\mathbb{E}\{|s_N|^2\}=\mathbb{E}\{|s_F|^2\}=1$; $a_N$ and $a_F$ are respectively the power coefficients allocated to $s_N$ and $s_F$ with $a_N+a_F=1$ and $a_F>a_N>0$. 

At the same time, the BD backscatters the BS’s signal $s$ to users with its own message $s_C$, which is intended for $U_N$, with $\mathbb{E}\{|s_C|^2\}=1$. It is worth noting that the additive noise at BD can be disregarded, as its integrated circuit only consists  of  passive components~\cite{long2019symbiotic}.  Hence, 
the signals received by $U_N$, $U_F$, and $\text{Eve-}V$ are respectively
given by
	\begin{align}\label{yN}
		y_N=\underbrace{h_{SN}s}_{\text{direct link}}+\underbrace{h_{SB}h_{BN}\beta ss_C}_{\text{backscatter link}}+n_N,
	\end{align}
	\begin{align}\label{yF}
		y_F=h_{SF}s+h_{SB}h_{BF}\beta ss_C+n_F,
	\end{align}
and 
\begin{align}\label{yV}
y_V=h_{SV}s+h_{SB}h_{BV}\beta ss_C+n_V,
\end{align}
where $|\beta|\leq 1$ denotes the reflection efficiency of BD; $n_N$, $n_F$, and $n_V\sim\mathcal{CN}(0,\sigma^2)$ respectively represent the additive white
Gaussian noises (AWGNs) at $U_N$, $U_F$, and $\text{Eve-}V$.

\begin{figure*}[!t]
	\normalsize 
	\setcounter{MYtempeqncnt}{\value{equation}} 
	\setcounter{equation}{11}
	\begin{align}\label{exp_of_OPC}
		P_{{\rm out},C}\!=\!\begin{cases}
			1\!+\!\mathcal{A}_1\exp\left(\mathcal{A}_1\!-\!\mathcal{B}_1\right){\rm Ei}\left(-\mathcal{A}_1\right)\!+\!\mathcal{A}_2\sum_{i=1}^{n}\sqrt{1-t_i^2}\exp(-\mathcal{B}_2)K_0\left(\mathcal{C}_2\right),&\!\! \frac{a_F}{a_N}>\gamma_{N,th}^{s_F}\,\text{and}\, \gamma_{N,th}^{s_C}<\frac{\theta}{\eta(1-\theta)},\\
			1,& \!\!\text{otherwise},
		\end{cases}			
	\end{align}
	\setcounter{equation}{\value{MYtempeqncnt}}
	\hrulefill 
\end{figure*}
 \subsection{Signal Decoding Process}
Please note that AN $z$ is generated via a pseudo-random sequence, which is shared between the legitimate transceivers  \cite{lv2019secure,li2019secrecy}. As shown in (\ref{yN}) and (\ref{yF}), the received signals consist of direct and backscatter links. Since the legitimate receivers have prior knowledge of $z$ and can obtain  the CSI before data transmission, they can achieve complete AN cancellation in the direct link \cite{lv2017improving,li2022improving,gu2023act}. However, the backscattered AN experiences time and phase shifts due to BD processing, making it challenging for users to fully recover $z$ without channel training and tracking. Hence, we introduce an attenuation factor $\eta$ to quantify the effectiveness of backscattered AN cancellation, $0 \leq \eta \leq 1$ \cite{li2022improving,9745773}.

 Following the AmBC-NOMA protocol,  $U_N$  first decodes $s_F$, then $s_N$, and ultimately decodes $s_C$  by using SIC technique. Given the AN injection described above, 
the received SINRs of $s_F$, $s_N$, and $s_C$ at $U_N$ can be respectively written as
\begin{align}\label{define_of_gammaUN_sF}
\gamma_{U_N}^{s_F}=\frac{|h_{SN}|^2\theta\gamma a_F}{
			|h_{SN}|^2\theta\gamma a_N \!\!+\!\!|h_{SB}|^2|h_{BN}|^2\beta^2 \left[\theta\!\!+\!\!\eta(1\!\!-\!\!\theta)\right]\gamma\!\!+\!\!1},
\end{align}
	\begin{align}\label{define_of_gammaUN_sN}
\gamma_{U_N}^{s_N}=\frac{|h_{SN}|^2 \theta\gamma a_N}{|h_{SB}|^2|h_{BN}|^2\beta^2 \left[\theta\!\!+\!\!\eta(1-\theta)\right]\gamma+1},
	\end{align}
and
	\begin{align}\label{define_gammaUNSC}
		\gamma_{U_N}^{s_C}=\frac{|h_{SB}|^2|h_{BN}|^2\beta^2\theta\gamma}{|h_{SB}|^2|h_{BN}|^2\beta^2\eta(1-\theta)\gamma+1},
	\end{align}
where $\gamma=P/\sigma^2$ represents transmit SNR. Meanwhile, $U_F$ only decodes $s_F$ by treating other signals as noise, whose SINR can be written as
\begin{align}
	\gamma_{U_F}^{s_F}=\frac{|h_{SF}|^2\theta\gamma a_F}{
		|h_{SF}|^2\theta\gamma a_N \!\!+\!\!|h_{SB}|^2|h_{BF}|^2\beta^2 \left[\theta\!\!+\!\!\eta(1\!\!-\!\!\theta)\right]\gamma\!\!+\!\!1}.
\end{align}

Next, we switch to the SINR at Eves. For each Eve, although its CSI is unknown, we assume that its statistical CSI is available \cite{liu2017enhancing}.\footnote{This is a general assumption in PLS analysis, which has been used in many works, such as  \cite{goel2008guaranteeing,liu2017enhancing,wang2022resisting}.}
Since Eves are non-colluding, they individually decode the received signals. Specifically, each Eve can  sequentially decode  $s_F$, $s_N$, and $s_C$ through SIC. Note that AN $z$ injected by BS in (\ref{transmitted_signal})  can be safeguarded from being compromised by Eves through  regular changes to the seed of the random sequence generator. Consequently, without prior information about $z$, Eves can hardly perform effective interference cancellation \cite{li2022improving}. Hence, BS
 can
limit the SINRs at Eves by controlling power fraction $\theta$. 
In this way, the eavesdropping SINRs towards $s_F$, $s_N$, and $s_C$  at $\text{Eve-}V$ can be respectively denoted by 
\begin{align}
\hspace*{-0.1cm}	\gamma_{V}^{s_F}\!\!=\!\!\frac{|h_{SV}|^2\theta\gamma a_F}{
		|h_{SV}|^2\theta\gamma a_N\!\!+\!\!|h_{SV}|^2(1\!\!-\!\!\theta)\gamma\!\!+\!\!|h_{SB}|^2|h_{BV}|^2\beta^2\gamma \!\!+\!\!1},
\end{align}
\begin{align}
	\gamma_{V}^{s_N}=\frac{|h_{SV}|^2\theta\gamma a_N}{
		|h_{SV}|^2(1\!\!-\!\!\theta)\gamma\!\!+\!\!|h_{SB}|^2|h_{BV}|^2\beta^2\gamma \!\!+\!\!1},
\end{align}
and
\begin{align}
	\gamma_{V}^{s_C}=\frac{|h_{SB}|^2|h_{BV}|^2\beta^2\theta\gamma}{
		(1\!\!-\!\!\theta)\gamma\left[|h_{SV}|^2\!\!+\!\!|h_{SB}|^2|h_{BV}|^2\beta^2\right] \!\!+\!\!1}.
\end{align}
We assume that the most detrimental Eve for $s_F$ has the highest instantaneous SINR for it, which can be written as  $\gamma_{E}^{s_F}=\max_{V\in\Phi_{\text{E}}}\{\gamma_{V}^{s_F}\}$. Similarly, the instantaneous SINRs at the most detrimental Eve  for $s_N$ and $s_C$ can be respectively expressed  as  $\gamma_{E}^{s_N}=\max_{V\in\Phi_{\text{E}}}\{\gamma_{V}^{s_N}\}$ as well as $\gamma_{E}^{s_C}=\max_{V\in\Phi_{\text{E}}}\{\gamma_{V}^{s_C}\}$.
\section{Outage Probability Analysis}
OP is a key performance metric when users' target rates are based on their required quality of service.
Hence, in this section, we present the OP analysis of the considered system model.
Additionally, we explore the asymptotic expressions of OP as well as the diversity order  at high SNR to provide more insights.
\subsection{OP of $U_F$}
Recall that $U_F$ only needs to decode $s_F$, it experiences an outage only when it fails to decode $s_F$ successfully.  Then the OP of $U_F$ can be written as
\begin{align}
		\setcounter{equation}{12}
P_{\text{out},F}=\Pr\left(\gamma_{U_F}^{s_F}\leq\gamma_{F,th}^{s_F}\right),
\end{align}
where $\gamma_{F,th}^{s_F}$ is the SINR threshold of $s_F$ at $U_F$.
Specifically, the theoretical expression for OP of $U_F$ is summarized in the following theorem.
\begin{theorem} \label{theorem_OPF}
The OP of $U_F$ can be derived as
\begin{align}\label{exp_of_OPF}
\hspace*{-0.3cm}P_{{\rm out},F}\!\!=\!\!\begin{cases}
		1\!+\!\mathcal{A}_0\exp\left(\mathcal{A}_0\!-\!\mathcal{B}_0\right){\rm Ei}\left(-\mathcal{A}_0\right),&\!\! \frac{a_F}{a_N}>\gamma_{F,th}^{s_F},\\
		1,&\!\! \frac{a_F}{a_N}\leq\gamma_{F,th}^{s_F},
	\end{cases}			
\end{align}
where $\mathcal{A}_0$ and $\mathcal{B}_0$ have been defined in Table \ref{OP_parameters}.
\end{theorem}
\begin{proof}
See Appendix \ref{Proof_of_OPF}.
\end{proof}
\subsection{OP of $U_N$}
Unlike  $U_F$, an outage occurs at  $U_N$
 when it fails to decode either $s_F$ or $s_N$. Thus the OP of $U_N$ can be expressed as
\begin{align}\label{case_when_OPN_occur}
	P_{\text{out},N}=1-\Pr\left(\gamma_{U_N}^{s_F}\geq\gamma_{N,th}^{s_F},\gamma_{U_N}^{s_N}\geq\gamma_{N,th}^{s_N}\right),
\end{align}
where $\gamma_{N,th}^{s_F}$ and $\gamma_{N,th}^{s_N}$ are respectively the SINR thresholds of $s_F$ and $s_N$ at $U_N$. 
The following theorem presents the theoretical expression for OP of $U_N$.
\begin{theorem} \label{theorem_OPN}
	The OP of $U_N$ can be derived as
	\begin{align}\label{exp_of_OPN}
		\hspace*{-0.3cm}P_{{\rm out},N}\!\!=\!\!\begin{cases}
			1\!+\!\mathcal{A}_1\exp\left(\mathcal{A}_1\!-\!\mathcal{B}_1\right){\rm Ei}\left(-\mathcal{A}_1\right),&\!\!\frac{a_F}{a_N}>\gamma_{N,th}^{s_F},\\
			1,&\!\!\frac{a_F}{a_N}\leq\gamma_{N,th}^{s_F},
		\end{cases}			
	\end{align}
	where $\mathcal{A}_1$ and $\mathcal{B}_1$ have been defined in Table \ref{OP_parameters}.
\end{theorem}
\begin{proof}
	See Appendix \ref{Proof_of_theorem_OPN}.
\end{proof}

\begin{table}[t]
	\centering
	\caption{ Parameters for OP of different users}  
	\label{OP_parameters}
	\scalebox{1}{	\begin{tabular}{|c|c|}
			\hline\hline
			OP & Corresponding Parameters \\
			\hline\hline
			\multirow{2}{*}{$P_{{\rm out},F}$} & $ 
			\mathcal{A}_0=\frac{\theta\left(a_F-a_N\gamma_{F,th}^{s_F}\right)\lambda_{SF}}{\beta^2\gamma_{F,th}^{s_F}\left[\theta+\eta(1-\theta)\right]\lambda_{SB}\lambda_{BF}}
			$   \\  
			&  $\mathcal{B}_0=\frac{\gamma_{F,th}^{s_F}}{\theta\gamma\left(a_F-a_N\gamma_{F,th}^{s_F}\right)\lambda_{SF} }$ \\
			\hline
			\multirow{4}{*}{$P_{{\rm out},N}$} & $ 
			\mathcal{A}_1=\frac{\theta\lambda_{SN}}{\beta^2\Gamma\left[\theta+\eta(1-\theta)\right]\lambda_{SB}\lambda_{BN}}
			$         \\ [1ex]  
			& $\mathcal{B}_1=\frac{\Gamma}{\theta\gamma\lambda_{SN}}$ \\    
			& $\Gamma=\max\left(\frac{\gamma_{N,th}^{s_F}}{a_F-a_N\gamma_{N,th}^{s_F}},\frac{\gamma_{N,th}^{s_N}}{a_N}\right)$\\
			\hline
			\multirow{6}{*}{$P_{{\rm out},C}$} &  $\mathcal{A}_2=\frac{\pi\gamma_{N,th}^{s_C}\exp(-\frac{\Gamma}{\theta\gamma\lambda_{SN}})}{n\beta^2\gamma\left[\theta-\eta(1-\theta)\gamma_{N,th}^{s_C}\right]\lambda_{SB}\lambda_{BN}}$          \\  
			& $\mathcal{B}_2=\frac{\Gamma\left[\theta+\eta(1-\theta)\right]\gamma_{N,th}^{s_C}\left(t_i+1\right)}{2\gamma\left[\theta-\eta(1-\theta)\gamma_{N,th}^{s_C}\right]\theta\lambda_{SN}}$   \\  
			& $\mathcal{C}_2=2\sqrt{\frac{\gamma_{N,th}^{s_C}\left(t_i+1\right)}{2\lambda_{SB}\lambda_{BN}\beta^2\gamma\left[\theta-\eta(1-\theta)\gamma_{N,th}^{s_C}\right]}}$\\ 
			& $t_i=\cos\left(\frac{(2i-1)\pi}{2n}\right)$\\  
			\hline\hline
	\end{tabular}}
\end{table}

\subsection{OP of BD}
According to the principles of AmBC-NOMA, the BD signal $s_C$ is considered successfully decoded only when $U_N$ can sequentially decode $s_F$, $s_N$, and $s_C$. Therefore, the OP of BD can be written as
\begin{align}\label{cases_when_sc_cannot_decode}
	P_{\text{out},C}\!\!=\!\!1\!-\!\Pr\left(\gamma_{U_N}^{s_F}\geq\gamma_{N,th}^{s_F},\gamma_{U_N}^{s_N}\geq\gamma_{N,th}^{s_N},\gamma_{U_N}^{s_C}\geq\gamma_{N,th}^{s_C}\right),
\end{align}
where $\gamma_{N,th}^{s_F}$ and $\gamma_{N,th}^{s_N}$ have been defined in (\ref{case_when_OPN_occur});  $\gamma_{N,th}^{s_C}$ is the SINR threshold  of $s_C$  at $U_N$. 
Subsequently,  we can achieve the following theorem 
\begin{theorem} \label{theorem_OPC}
	The OP of BD can be derived as (\ref{exp_of_OPC}), shown at the top of this page,
	where $n$ represents the calculation accuracy; $\mathcal{A}_1$, $\mathcal{B}_1$, $\mathcal{A}_2$, $\mathcal{B}_2$, $\mathcal{C}_2$, and $t_i$  have been defined in Table \ref{OP_parameters}.
\end{theorem}
\begin{proof}
	See Appendix \ref{Proof_of_theorem_OPC}.
\end{proof} 
 \subsection{Asymptotic Analysis}
To provide important insights at high SNR, i.e., $\gamma\rightarrow\infty$, we explore the asymptotic outage behavior of different users, and then conduct the asymptotic diversity order analysis based on OP in the following corollaries.
\begin{corollary}\label{asy_OPF}
	At high SNR, the asymptotic OP of $U_F$ can be derived as
\begin{align}\label{outage_F_infty}
	P_{{\rm out},F}^{\text{asy}}=
	1\!+\!\mathcal{A}_0\exp\left(\mathcal{A}_0\right)(1-\mathcal{B}_0){\rm Ei}\left(-\mathcal{A}_0\right),
\end{align}	
when $\frac{a_F}{a_N}>\gamma_{F,th}^{s_F}$.	 
	
\end{corollary} 
 \begin{proof}
The above equation can be straightforwardly attained by using the Taylor series expansion of $\exp(\cdot)$, where $\exp(-x) \approx 1 - x$ when $x \to 0$.
 \end{proof}
\begin{corollary}\label{asy_OPN}
	At high SNR, the asymptotic OP of $U_N$ can be derived as
	\begin{align}\label{outage_N_infty}
		P_{{\rm out},N}^{\text{asy}}=
		1\!+\!\mathcal{A}_1\exp\left(\mathcal{A}_1\right)(1\!-\!\mathcal{B}_1){\rm Ei}\left(-\mathcal{A}_1\right),
	\end{align}
	when $\frac{a_F}{a_N}>\gamma_{N,th}^{s_F}$.
\end{corollary} 
\begin{corollary}\label{asy_OPC}
	At high SNR, the asymptotic OP of BD can be derived as
	\begin{align}\label{outage_C_infty}
	&	P_{{\rm out},C}^{\text{asy}}=
		1\!+\!\mathcal{A}_1\exp\left(\mathcal{A}_1\right)(1\!-\!\mathcal{B}_1){\rm Ei}\left(-\mathcal{A}_1\right)\nonumber\\
		&\!-\!	\mathcal{A}_2^{\text{asy}}\frac{(1-\frac{\Gamma}{\theta\gamma\lambda_{SN}})}{\gamma}\sum_{i=1}^{n}\sqrt{1-t_i^2}(1-\mathcal{B}_2)\ln\left(\frac{\mathcal{C}_2}{2}\right),
	\end{align}
	where 
	\begin{align}
	\mathcal{A}_2^{\text{asy}}=\frac{\pi\gamma_{N,th}^{s_C}}{n\beta^2\left[\theta-\eta(1-\theta)\gamma_{N,th}^{s_C}\right]\lambda_{SB}\lambda_{BN}},
	\end{align}
	when $\frac{a_F}{a_N}>\gamma_{N,th}^{s_F}\,\text{and}\, \gamma_{N,th}^{s_C}<\frac{\theta}{\eta(1-\theta)}$.
\end{corollary}
 \begin{proof}
 The above equation can be derived by using $\exp(-x) \approx 1 - x$ and $K_0(x)\approx-\ln(\frac{x}{2})$, when $x \to 0$.
 \end{proof}
 \begin{corollary}\label{OP_infty}
From corollaries \ref{asy_OPF}, \ref{asy_OPN}, and \ref{asy_OPC}, we can respectively attain the lower floors of $P_{{\rm out},F}$, $P_{{\rm out},N}$, and $P_{{\rm out},C}$ as 
\begin{align}\label{OP_infty_F}
	P_{{\rm out},F}^{\infty}=
	1\!+\!\mathcal{A}_0\exp\left(\mathcal{A}_0\right){\rm Ei}\left(-\mathcal{A}_0\right),
\end{align}	
and
\begin{align}\label{OP_infty_N_C}
	P_{{\rm out},N}^{\infty}=P_{{\rm out},C}^{\infty}=
	1\!+\!\mathcal{A}_1\exp\left(\mathcal{A}_1\right){\rm Ei}\left(-\mathcal{A}_1\right).
\end{align}
 \end{corollary}
  \begin{corollary}\label{diversity_of_ops}
The diversity order based on OP can be expressed as \cite{zheng2003diversity}:
\begin{align}\label{define_of_diversity}
d_i=-\lim\limits_{\gamma\rightarrow\infty }\frac{\log P_{{\rm out},i}^{\infty}}{\log\gamma},
\end{align}
where $i\in\{F,N,C\}$.
By substituting (\ref{OP_infty_F}) and (\ref{OP_infty_N_C}) into (\ref{define_of_diversity}), we can obtain the following relationship:
\begin{align}
d_F=d_N=d_C=0.
\end{align}
  \end{corollary}
   \begin{proof}
  	The above corollary can be easily obtained from the results presented in \textbf{Corollary  \ref{OP_infty}}.
  \end{proof}
\begin{remark}\label{remark1}
{\rm Users' SINR  is impacted by inter-user interference and residual AN. As the transmit SNR increases, the users' SINR correspondingly improves, pushing the system's outage performance to its lower floor, with the diversity order eventually converging to zero. These lower floors are influenced by the power allocation coefficient $a_N$ (or $a_F$), the power fraction allocated to the useful signal $\theta$, backscattered AN cancellation efficiency $\eta$, and the reflection efficiency $\beta$.  Furthermore, an interesting insight is gained: $P_{{\rm out},N}^{\infty} = P_{{\rm out},C}^{\infty}$. This implies that, at high SNR, the SINR threshold of $s_C$ at $U_N$ becomes negligible.}

\end{remark}
\begin{figure*}[!t]
	\normalsize 
	\setcounter{MYtempeqncnt}{\value{equation}} 
	\setcounter{equation}{27}
	\begin{align}\label{exp_of_IPC}
		P_{\text{int},C}
		\!=\!\begin{cases}
			1\!-\!\exp\left\{\!-2\pi\lambda_e\sum_{i=1}^{N}w_i\left[1 \!+\! \mathcal{G}_{e2}\exp(\mathcal{B}_{e2})\right.\right. \\ 
			\left.\left. \quad\quad\quad\quad    - \mathcal{A}_{e3} \sum_{i=1}^{n} \sqrt{1 \!-\! t_i^2} \left(1 \!-\! \exp(-\mathcal{B}_{e3})\right) K_0(\mathcal{C}_{e3})\right](l_i\!+\!r_p)\exp(l_i)\right\}, &\!\! \frac{\theta}{1-\theta}>\gamma_{E,th}^{s_C}, \\[5pt]
			0, & \!\!\frac{\theta}{1-\theta}\leq\gamma_{E,th}^{s_C},
		\end{cases}	
	\end{align}
	\begin{align}\label{IP_BD_asy}
		P_{\text{int},C}^{\textit{asy}}
		\!=\!
		1\!-\!\exp\left\{\!-2\pi\lambda_e\sum_{i=1}^{N}w_i\left[1 \!+\! \mathcal{G}_{e2}(1+\mathcal{B}_{e2})+ \mathcal{A}_{e3} \sum_{i=1}^{n} \sqrt{1 \!-\! t_i^2} \ln\left(\frac{\mathcal{C}_{e3}}{2}\right)\mathcal{B}_{e3}\right](l_i\!+\!r_p)\exp(l_i)\right\}
	\end{align}
	\setcounter{equation}{\value{MYtempeqncnt}}
	\hrulefill 
\end{figure*} 
\begin{table}[t]
	\centering
	\caption{Parameters for IP of different users}  
	\label{IP_parameters}
	\scalebox{1}{
		\begin{tabular}{|c|c|}
			\hline\hline
			IP & Corresponding Parameters \\
			\hline\hline
			\multirow{3}{*}{$P_{\text{int},F}$} & $ 
			\mathcal{A}_{e0}=\frac{[\theta a_F-\theta a_N\gamma_{E,th}^{s_F}-(1-\theta) \gamma_{E,th}^{s_F}]}{\lambda_{SB}\gamma_{E,th}^{s_F}\beta^2}
			$  \\  
			& $\mathcal{G}_{e0}=\mathcal{A}_{e0}\exp(\mathcal{A}_{e0})\text{Ei}(-\mathcal{A}_{e0})$ \\  
			& $\mathcal{M}_{e0}=\frac{\gamma_{E,th}^{s_F}}{[\theta a_F-\theta  a_N\gamma_{E,th}^{s_F}-(1-\theta)  \gamma_{E,th}^{s_F}]\gamma}$ \\ 
			\hline
			\multirow{4}{*}{$P_{{\rm int},N}$} & $ 
			\mathcal{A}_{e1}=\frac{[\theta a_N-(1-\theta)\gamma_{E,th}^{s_N}]}{\lambda_{SB}\gamma_{E,th}^{s_N}\beta^2}
			$  \\  
			& $\mathcal{G}_{e1}=\mathcal{A}_{e1}\exp(\mathcal{A}_{e1})\text{Ei}(-\mathcal{A}_{e1})$ \\  
			& $\mathcal{M}_{e1}=\frac{\gamma_{E,th}^{s_N}}{\gamma[\theta a_N-(1-\theta)\gamma_{E,th}^{s_N}]}$ \\ 
			\hline
			\multirow{11}{*}{$P_{{\rm int},C}$} & $\mathcal{A}_{e2}= \frac{\gamma_{E,th}^{s_C} (1 - \theta)}{\lambda_{SB}\delta}$ \\  
			& $\mathcal{B}_{e2} = \frac{(l_i+r_p)^\alpha}{ (1 - \theta) \gamma}$ \\  
			& $\mathcal{G}_{e2}=\mathcal{A}_{e2}\exp(\mathcal{A}_{e2})\text{Ei}(-\mathcal{A}_{e2})$ \\  
			& $\mathcal{A}_{e3}=\frac{\gamma_{E,th}^{s_C}\pi(l_i+r_p)^\alpha}{2\gamma\delta\lambda_{SB}}$\\
			&$\mathcal{B}_{e3}=\frac{\gamma_{E,th}^{s_C}(t_i+1)\delta(l_i+r_p)^\alpha}{\gamma\delta2\gamma_{E,th}^{s_C}(1-\theta)}-\frac{(l_i+r_p)^\alpha}{(1-\theta)\gamma}$\\
			& $\mathcal{C}_{e3}=2\sqrt{\frac{\gamma_{E,th}^{s_C}(t_i+1)(l_i+r_p)^\alpha}{\gamma\delta2\lambda_{SB}}}$ \\ 
			& $\delta= \beta^2 [\theta - (1 - \theta)\gamma_{E,th}^{s_C}]$ \\ 
			& $w_i=\frac{l_i}{(N+1)^2[L_{N+1}(l_i)]^2}$ \\ 
			& $t_i=\cos\left(\frac{(2i-1)\pi}{2n}\right)$\\
			\hline\hline
		\end{tabular}
	}
\end{table}
\section{Intercept Probability Analysis}
An intercept event is considered to occur once the secrecy rate becomes negative. Therefore, IP, defined as the probability that Eve successfully intercepts the legitimate signal, is a key metric for evaluating PLS performance \cite{zou2013optimal}. In this section, we examine the system's intercept performance by deriving theoretical expressions for the IP, along with its asymptotic solution in the high SNR regime.
 
  In accordance with the definition of IP, legitimate users or the BD are considered intercepted if the most detrimental Eve successfully decodes their corresponding messages. Let  $\gamma_{E,th}^{s_i}$ ($i \in \{N,F,C\}$) represent the secrecy SINR threshold of $s_i$, then 
the IP of the most detrimental Eve intercepting signal $s_i$  is given by
\begin{align}\label{IP_define}
	P_{\text{int},i}&=\Pr\left(\gamma_{E}^{s_i}>\gamma_{E,th}^{s_i}\right)=1-\Pr\left(\max_{V\in\Phi_{\text{E}}}\{\gamma_{V}^{s_i}\}\leq\gamma_{E,th}^{s_i}\right)\nonumber\\
 	&=1-\mathbb{E}_{\Phi_{\text{E}}}\left\{\prod\limits_{v\in\Phi_{\text{E}}\atop d_{SV}\geq r_p}\Pr\left(\gamma_{V}^{s_F}\leq \gamma_{E,th}^{s_i}\right)\right\}\nonumber\\
	&\overset{(a)}{=}1-\mathbb{E}_{\Phi_{\text{E}}}\left\{\prod\limits_{v\in\Phi_{\text{E}}\atop d_{SV}\geq r_p}F_{\gamma_{V}^{s_i}}(\gamma_{E,th}^{s_i})\right\}\nonumber\\
	&\overset{(b)}{=}1-\exp\left(-\lambda_e\int_{S}\left[1-F_{\gamma_{V}^{s_i}}(\gamma_{E,th}^{s_i})\right]dd_{SV}\right),
\end{align}
where (a) follows from the independent nature of the non-colluding Eves' SINRs, and (b) is based on the assumption that $d_{SV}\approx d_{BV}$, following from the probability generating functional (PGFL) of an HPPP. In this manner, we can achieve the following theorems.

\subsection{IP of $U_F$}
Based on (\ref{IP_define}), we can provide the theoretical expression for the IP of $U_F$ in the following theorem.
\begin{theorem}\label{theorem_IPF}
	The IP of $U_F$ can be derived as
	\begin{align}\label{IP_UF}
		P_{{\rm int},F}\!\!=\!\!\begin{cases}
			1\!-\!\exp\left(\lambda_e\mathcal{G}_{e0}2\pi\frac{\Gamma(\frac{2}{\alpha},\mathcal{M}_{e0}r_p^\alpha)}{\alpha\mathcal{M}_{e0}}\right),&\!\!\!\frac{\theta a_F}{\theta a_N+1-\theta}>\gamma_{E,th}^{s_F},\\
			0,&\!\!\!\frac{\theta a_F}{\theta a_N+1-\theta}\leq\gamma_{E,th}^{s_F},
		\end{cases}			
	\end{align}
	where $\mathcal{G}_{e0}$ and $\mathcal{M}_{e0}$ have been defined in Table \ref{IP_parameters}.
\end{theorem}
\begin{proof}
	See Appendix \ref{Proof_of_theorem_IPF}.	
\end{proof}

\subsection{IP of $U_N$}
Similar to \textbf{Theorem \ref{theorem_IPF}}, the IP of $U_N$ can be obtained and is outlined in the following theorem.
\begin{theorem}
The IP of $U_N$ can be derived as
\begin{align}\label{IP_UN}
	\setcounter{equation}{29}
	P_{{\rm int},N}\!\!=\!\!\begin{cases}
		1\!-\!\exp\left(\lambda_e\mathcal{G}_{e1}2\pi\frac{\Gamma(\frac{2}{\alpha},\mathcal{M}_{e1}r_p^\alpha)}{\alpha\mathcal{M}_{e1}}\right),&\!\!\! \frac{\theta a_N}{1-\theta}>\gamma_{E,th}^{s_N},\\
		0,&\!\!\! \frac{\theta a_N}{1-\theta}\leq\gamma_{E,th}^{s_N},
	\end{cases}			
\end{align}	
where $\mathcal{G}_{e1}$ and $\mathcal{M}_{e1}$ have been defined in Table \ref{IP_parameters}.
\end{theorem}
\begin{proof}
	The proof is similar to Appendix \ref{Proof_of_theorem_IPF}, and thus omitted here due to space limitation.
\end{proof}
\subsection{IP of BD}
\begin{theorem}\label{IP_BD}
The IP of BD can be derived as (\ref{exp_of_IPC}), shown at the top of this page,
where $n$ and $N$ represent the calculation accuracy; $l_i$ is the $i$-th root of the Laguerre polynomial, $L_{N}(l)$;    $\mathcal{A}_{e3}$, $\mathcal{B}_{e2}$, $\mathcal{B}_{e3}$, $\mathcal{C}_{e3}$, $\mathcal{G}_{e2}$, 
 $t_i$, and $w_i$ have been defined in Table~\ref{IP_parameters}.
\end{theorem}
\begin{proof}
See Appendix \ref{Proof_of_theorem_IPC}.
\end{proof}
 \subsection{Asymptotic Analysis}
The intercept performance of Eves at high SNR plays a crucial role for effective PLS design. Accordingly, in this subsection, we present the asymptotic expressions for the IP in the high SNR regime, leading to the following corollaries.

\begin{corollary}\label{asy_IP_UF}	
	At high SNR, the asymptotic IP of $U_F$ can be derived as
\begin{align}
P_{{\rm int},F}^{\text{asy}}\!\!=\!\!
	1\!\!-\!\!\exp\!\left(\lambda_e\mathcal{G}_{e0}2\pi\frac{\Gamma\left(\frac{2}{\alpha}\right) \!\!-\!\! \sum_{n=0}^{\infty}  \frac{(-1)^n\left(\mathcal{M}_{e0} r_p^\alpha\right)^{\frac{2}{\alpha} + n}}{n!(\frac{2}{\alpha}+n)}}{\alpha\mathcal{M}_{e0}}\right),
\end{align}
when $\frac{\theta a_F}{\theta a_N+1-\theta}>\gamma_{E,th}^{s_F}$.
\end{corollary}
\begin{proof}
As $\gamma \rightarrow \infty$, we observe that $\mathcal{M}_{e0} \rightarrow 0^{+}$. Since $\frac{2}{\alpha} > 0$, we can employ the asymptotic series expansion \cite{bender2013advanced}:
\begin{align}
	\Gamma\left(\frac{2}{\alpha}, \mathcal{M}_{e0} r_p^\alpha\right) \approx \Gamma\left(\frac{2}{\alpha}\right) \!-\! \sum_{n=0}^{\infty} (-1)^n \frac{\left(\mathcal{M}_{e0} r_p^\alpha\right)^{\frac{2}{\alpha} + n}}{n!(\frac{2}{\alpha}+n)},
\end{align}
completing the proof.	
\end{proof}	

\begin{corollary}
	At high SNR, the asymptotic IP of $U_N$ can be derived as
	\begin{align}
		P_{{\rm int},N}^{\text{asy}}\!\!=\!\!
		1\!\!-\!\!\exp\!\left(\lambda_e\mathcal{G}_{e1}2\pi\frac{\Gamma\left(\frac{2}{\alpha}\right) \!\!-\!\! \sum_{n=0}^{\infty}  \frac{(-1)^n\left(\mathcal{M}_{e1} r_p^\alpha\right)^{\frac{2}{\alpha} + n}}{n!(\frac{2}{\alpha}+n)}}{\alpha\mathcal{M}_{e1}}\right),
	\end{align}
	when
	$\frac{\theta a_N}{1-\theta}>\gamma_{E,th}^{s_N}$.
\end{corollary}

\begin{corollary}\label{asy_IP_BD}
	At high SNR, when $\frac{\theta}{1-\theta}>\gamma_{E,th}^{s_C}$, the asymptotic IP of BD can be derived as (\ref{IP_BD_asy}), shown at the top of this page.
\end{corollary}
\begin{proof}
 The above equation can be derived by using $\exp(-x) \approx 1 - x$, $\exp(x) \approx 1 + x$, and $K_0(x)\approx-\ln(\frac{x}{2})$, when $x \to 0$.
\end{proof}

\begin{remark}\label{remark2}
{\rm	From theorems \ref{theorem_IPF} to \ref{IP_BD}, and corollaries \ref{asy_IP_UF} to \ref{asy_IP_BD}, we can find the following observations: 
Unlike OP, IP cannot be affected by the backscattered AN cancellation efficiency $\eta$, as Eve lacks the capability to cancel AN. However, IP is dependent on parameters such as the power allocation coefficient $a_N$ (or $a_F$), the power fraction $\theta$, and the reflection efficiency $\beta$. Additionally, it can be affected by the Eve distribution density $\lambda_e$ and the radius of  Eve-exclusion area $r_p$. 
Recall that the exponent terms in $P_{{\rm int},F}$, $P_{{\rm int},N}$, and $P_{{\rm int},C}$ can be given by $\exp(\lambda_e \cdot \mathcal{C})$, where $\mathcal{C}$ is a negative constant. As $\lambda_e$ increases, $\exp(\lambda_e \cdot\mathcal{C})$ decreases, causing IP to approach  1. Consequently, enlarging $r_p$ becomes especially critical in PLS design.	}
\end{remark}

\section{Numerical Results}
In this section, we confirm the accuracy of  theoretical expressions through Monte Carlo simulations. Unless otherwise specified, the simulation is configured with the following parameters  \cite{liu2016cooperative,liu2017enhancing,zhang2019backscatter,li2021hardware,li2021cognitive,li2022improving,pei2024secrecy,peiwcl}: Power allocation coefficients $a_N = 0.2$ and $a_F = 0.8$; reflection efficiency $\beta = 0.5$; power fraction $\theta = 0.9$; attenuation factor $\eta = 0.5$; distances between legitimate nodes: $d_{SF} = d_{SB} = 2$, $d_{SN} = 1$, $d_{BN} = 1.1$, $d_{BF} = 1.5$; path loss exponent $\alpha = 2$; SINR thresholds: $\gamma_{F,th}^{s_F} = \gamma_{N,th}^{s_F} = \gamma_{N,th}^{s_N} = 0.1$, $\gamma_{N,th}^{s_C} = 0.05$, $\gamma_{E,th}^{s_i} = 0.1$, where $i\in\{N,F,C\}$; Eve distribution density $\lambda_e=10^{-4}$;
radius of the Eve-exclusion area $r_p = 10$~m and the outer radius of the Eve area is 1000 m.\footnote{Recall that Eves are randomly distributed in the Eve area (an annular region extending from $r_p$ to $\infty$). However, $\infty$ is impractical for simulation. Therefore, following \cite{liu2017enhancing}, we replace $\infty$ with 1000 m. This approximation is reasonable since 1000 m is much larger than $r_p$ and can be considered as infinite.}

\begin{figure}[t]
	\centering
	\includegraphics[width=3.5in]{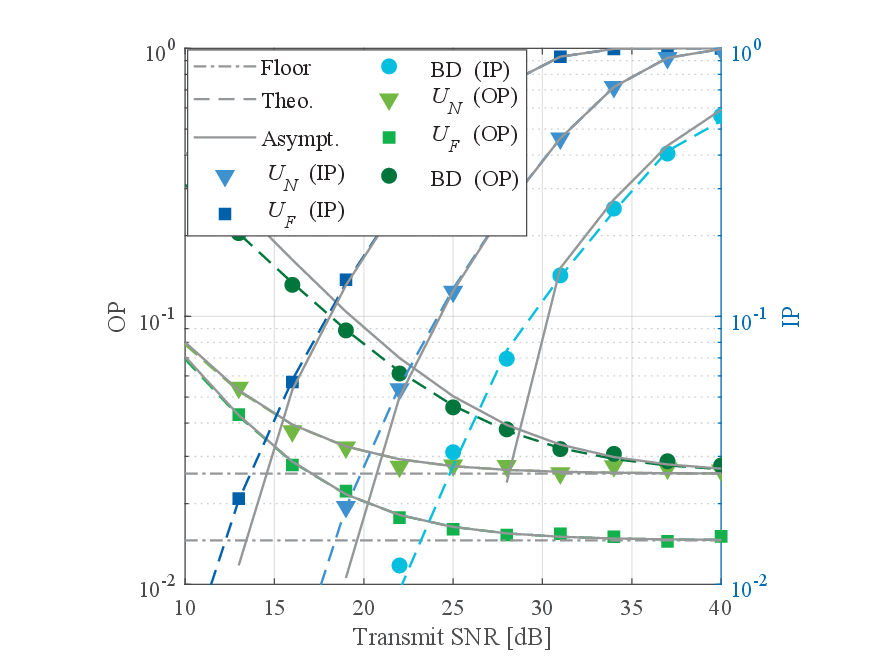}
	\caption{OP and IP versus transmit SNR.} 	
	\label{IP_OP_VS_SNR}
\end{figure}

In Fig. \ref{IP_OP_VS_SNR}, we portray the IP and OP for different users with respect to the transmit SNR $\gamma$ at BS. Clearly, the theoretical results align well with the simulation ones, and the asymptotic results closely approximate the exact ones in the high SNR regime, demonstrating their accuracy. It can be observed that all the IP curves steadily increase as $\gamma$ rises. In contrast, the OP curves decrease with increasing $\gamma$. However, at high SNR, the OP converges to fixed non-zero floor values, which ultimately results in zero diversity order. Moreover, we can observe that the OPs of $U_N$ and BD converge to the same floor at high SNR, verifying \textbf{Remark \ref{remark1}}. 
On the other hand, it can be observed that the OP of $U_F$ exhibits the best performance, followed by $U_N$, while BD shows the worst performance. This behavior comes from the different SIC processes for the two users and BD.
Similarly, the IP performance of $U_N$, $U_F$, and BD follows the same order, which is determined by the differences in power allocation coefficients and channel characteristics.

\begin{figure}[t]
	\centering
	\includegraphics[width=3.5in]{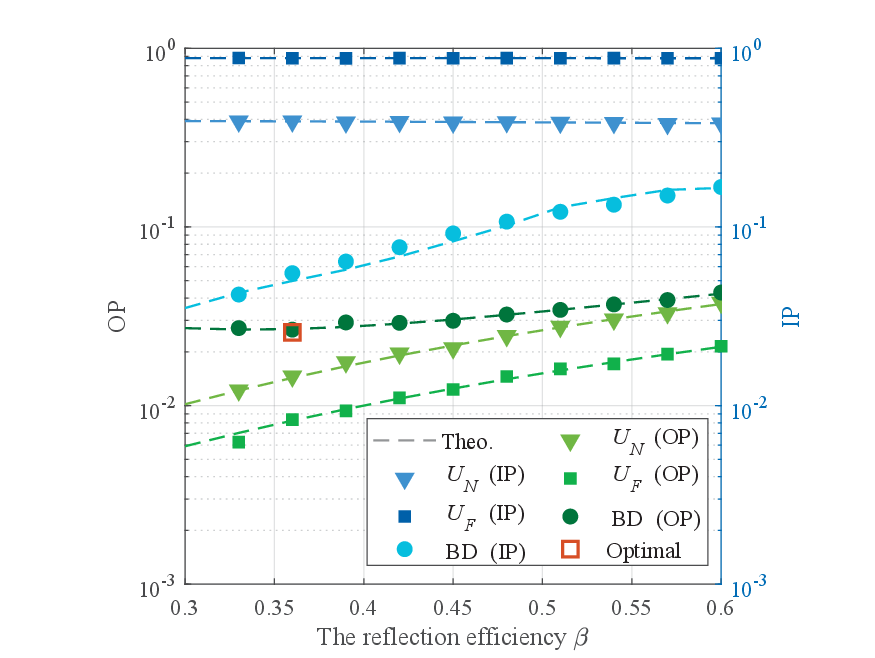}
	\caption{OP and IP versus reflection efficiency.} 	
	\label{IP_OP_VS_beta}
\end{figure}
Fig. \ref{IP_OP_VS_beta} depicts the dependence between IP and OP versus the reflection efficiency $\beta$ when transmit SNR $\gamma=30$ dB. It can be observed that the curves of OP for $U_N$ and $U_F$ increase with increasing $\beta$, while the curves of IP decrease. This happens due to the greater interference caused by the backscatter link as $\beta$ increases. On the other hand, the IP of BD increases as $\beta$ increases, due to the fact that the stronger backscatter link enhances the eavesdropping SINR towards  $s_C$. The OP of BD, however, initially decreases to an optimal point  and then increases. 
Recall that the OP of BD is jointly affected by the SINRs of $s_N$, $s_F$, and $s_C$. Referring to (\ref{define_of_gammaUN_sF}), (\ref{define_of_gammaUN_sN}), and (\ref{define_gammaUNSC}), one can find that the SINRs of $s_F$ and $s_N$ decrease, while the SINR of $s_C$ increases  with stronger backscatter link. When $\beta$ is small, the increase in the SINR of $s_C$ plays a dominant role, while as $\beta$ increases, the reduction in the SINRs of $s_N$ and $s_F$ becomes non-negligible, leading to increasing OP. Therefore, selecting an appropriate $\beta$ is essential for PLS design.

\begin{figure}[t]
	\centering
	\includegraphics[width=3.5in]{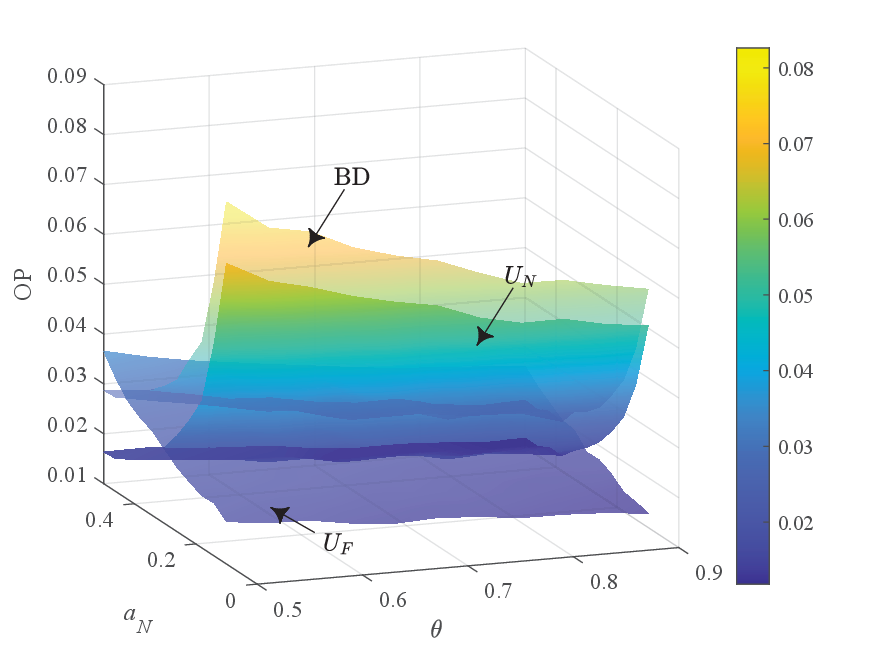}
	\caption{OP versus power allocation coefficient $a_N$ and power fraction $\theta$.} 	
	\label{OP_VS_aN_and_theta}
\end{figure}

The OPs of different users versus power allocation coefficient $a_N$ and power fraction $\theta$ when transmit SNR $\gamma=30$ dB is plotted
in Fig. \ref{OP_VS_aN_and_theta}. It can be observed that with increasing $\theta$, the OP decreases. This is because a larger $\theta$ indicates that more power is allocated to the desired signal, thereby improving the outage performance of the system. Then we switch to the relationship of OP and $a_N$. Clearly, allocating more power to $U_N$ significantly improves its outage performance, but it results in a deterioration in the OP of $U_F$.
On the other hand, the outage performance of BD is jointly affected by those of $U_N$ and $U_F$. As shown in the figure, the performance trend of BD follows that of $U_N$, indicating that the outage performance of $U_N$ has a greater impact on BD than that of $U_F$.

\begin{figure}[t]
	\centering
	\includegraphics[width=3.5in]{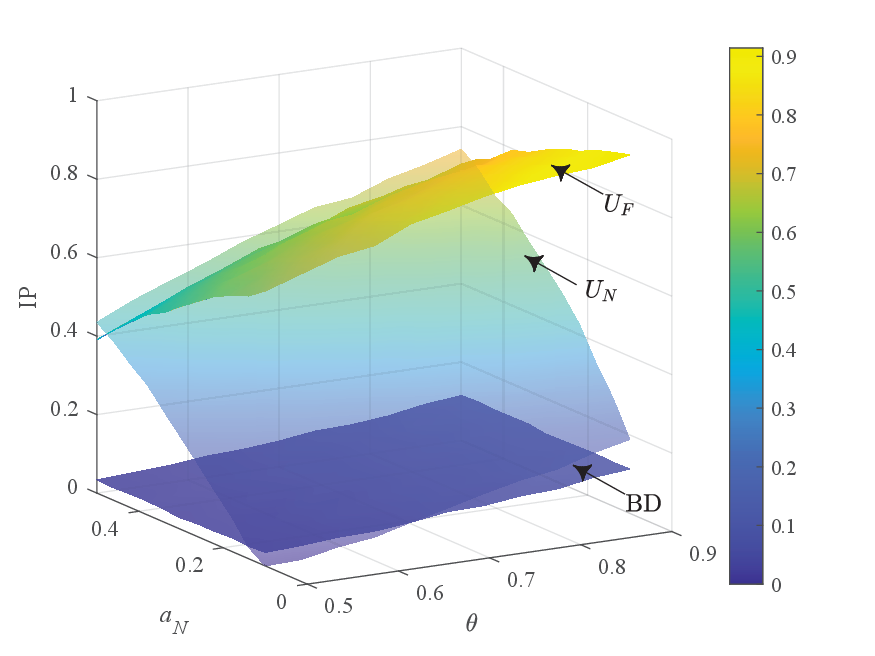}
	\caption{IP versus power allocation coefficient $a_N$ and power fraction $\theta$.} 	
	\label{IP_VS_aN_and_theta}
\end{figure}

In Fig. \ref{IP_VS_aN_and_theta}, the IP versus power allocation coefficient $a_N$ and power fraction $\theta$ when transmit SNR $\gamma=30$ dB is illustrated.  In contrast to Fig. \ref{OP_VS_aN_and_theta}, the IP curves increase monotonically with $\theta$, indicating that larger $\theta$ can degrade the secrecy performance. Recall that an increment in $\theta$ also enhances the outage performance. Thus, selecting an appropriate value of $\theta$ is crucial to achieve an optimal trade-off between security and reliability.
Next, we analyze the effect of $a_N$. It is evident that the IP of $U_F$ decreases as $a_N$ increases, whereas the IP of $U_N$ increases. Comparing Figs. \ref{OP_VS_aN_and_theta} and \ref{IP_VS_aN_and_theta}, we find that the corresponding IP and OP curves exhibit entirely opposite trends with respect to $a_N$.
Finally, the IP of BD remains unchanged, as it is independent of the power allocation coefficient.


\begin{figure}[t]
	\centering
	\includegraphics[width=3.5in]{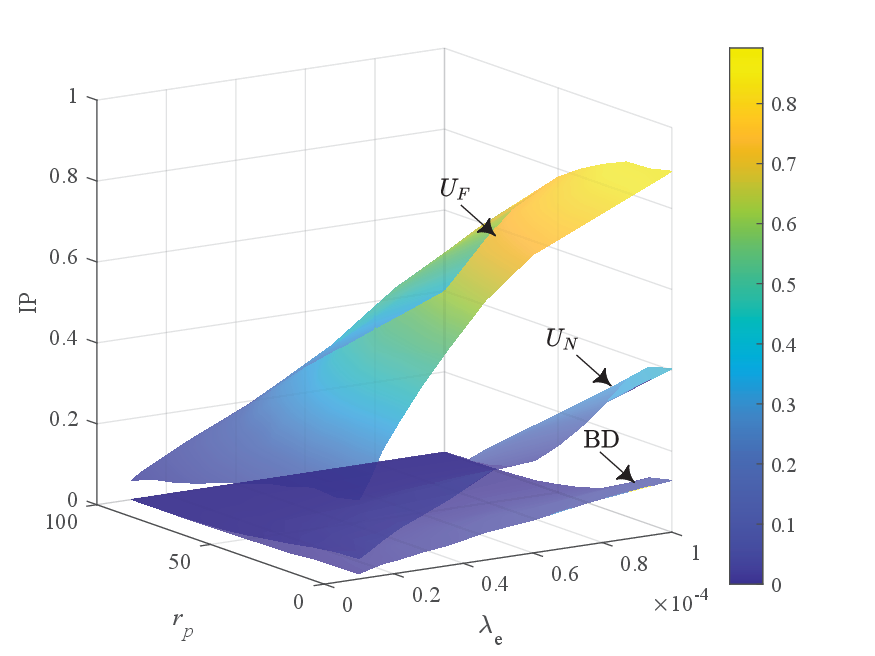}
	\caption{IP versus Eve-exclusion area radius $r_p$ and Eve distribution density $\lambda_e$.} 	
	\label{IP_VS_rp_and_lambda}
\end{figure}
Fig. \ref{IP_VS_rp_and_lambda} illustrates the impact of the Eve-exclusion area radius $r_p$ and Eve distribution density $\lambda_e$ on the IP when  transmit SNR $\gamma=30$ dB. Evidently, the IP decreases as $r_p$ increases. This phenomenon indicates that one effective approach to enhancing PLS is to enlarge the Eve-exclusion area. As expected in \textbf{Remark \ref{remark2}}, a lower Eve density $\lambda_e$ results in a reduced IP, demonstrating improved  secrecy performance. This is because a lower $\lambda_e$ results in fewer Eves, thereby weakening the multi-user diversity gain associated with the selection of the most detrimental Eve. Consequently, the harm inflicted by the most detrimental Eve is mitigated, reducing the IP.

\begin{figure}[t]
	\centering
	\includegraphics[width=3.5in]{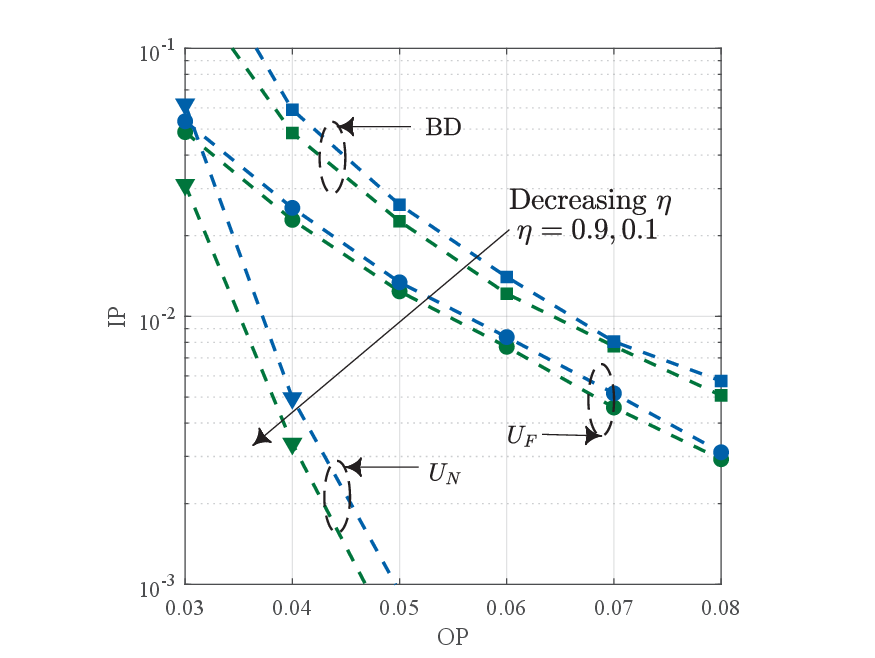}
	\caption{IP versus OP with different attenuation factor $\eta$.} 	
	\label{IP_VS_OP_with_eta}
\end{figure}

Fig. \ref{IP_VS_OP_with_eta} shows the impact of IP versus OP with different attenuation factor $\eta$. The results reveal an inverse relationship between OP and IP: as one increases, the other decreases. This confirms the trade-off between outage and intercept performances. Furthermore, it is evident that a smaller $\eta$ improves the reliability-security trade-off of the system, although marginally. This is because $\eta$ reflects the ability of legitimate nodes to cancel backscattered AN, while Eves struggle to mitigate it. Therefore, the effect of $\eta$ on the overall system performance is limited. Thus, to achieve an optimal security-reliability trade-off in the proposed system, more focus should be paid on the design of other parameters.

\section{Conclusions}
In this work, we have investigated an AmBC-NOMA system comprising a BS, a BD, two users, and multiple non-colluding passive HPPP-distributed Eves. The derived theoretical expressions for OP and IP have effectively characterized the system's PLS performance. We have also investigated the asymptotic OP and IP to better understand the system's behavior at high SNR.
Simulation results have validated the correctness of theoretical results. Besides, the effects of various factors, including transmit SNR, reflection efficiency, power allocation coefficient, power fraction, Eve-exclusion area radius, Eve distribution density, and attenuation factor, have also been observed.  
Finally, simulations have demonstrated that incorporating AN injection and establishing an Eve-exclusion area can significantly enhance the reliability and security of the AmBC-NOMA system.

\begin{appendices}
	\setcounter{equation}{0}
	\renewcommand{\theequation}{\thesection.\arabic{equation}}

\begin{figure*}[!t]
	\normalsize 
	\setcounter{MYtempeqncnt}{\value{equation}} 
	\begin{align}\label{first_simple_of_PoutC}
		P_{\text{out},C}&=1-\Pr\left(|h_{SN}|^2\geq \nu_N\max\left(\frac{\kappa_3}{\kappa_4},\frac{\kappa_5}{\kappa_6}\right)+\max\left(\frac{\gamma_{N,th}^{s_F}}{\kappa_4},\frac{\gamma_{N,th}^{s_N}}{\kappa_6}\right),\nu_N>\frac{\gamma_{N,th}^{s_C}}{\kappa_7}\right)\nonumber\\
		&=1-\underbrace{\Pr\left(|h_{SN}|^2\geq \nu_N\frac{\beta^2\Gamma\left[\theta+\eta(1-\theta)\right]}{\theta}+\frac{\Gamma}{\theta\gamma},\nu_N>\frac{\gamma_{N,th}^{s_C}}{\kappa_7}\right)}_{\mathcal{Q}}=1-\mathcal{Q}\nonumber		\tag{C.3}
	\end{align}
	\begin{align}\label{compose_of_Q}
		\mathcal{Q}=&\int_{\frac{\gamma_{N,th}^{s_C}}{\kappa_7}}^{+\infty}\int_{v\kappa_8+\frac{\Gamma}{\theta\gamma}}^{+\infty}f_{|h_{SN}|^2}(x)f_{\nu_N}(v)dxdv\nonumber\\
		=&\int_{\frac{\gamma_{N,th}^{s_C}}{\kappa_7}}^{+\infty}\exp\left(-v\frac{\kappa_8}{\lambda_{SN}}-\frac{\Gamma}{\lambda_{SN}\theta\gamma}\right)f_{\nu_N}(v)dv\nonumber\\
		=&\underbrace{\int_{0}^{+\infty}\exp\left(-v\frac{\kappa_8}{\lambda_{SN}}-\frac{\Gamma}{\lambda_{SN}\theta\gamma}\right)\frac{2}{\lambda_{SB} \lambda_{BN}} K_0\left(2\sqrt{\frac{v}{\lambda_{SB} \lambda_{BN}}}\right)dv}_{\mathcal{I}_F}\nonumber\\
		&-\underbrace{\int_{0}^{\frac{\gamma_{N,th}^{s_C}}{\kappa_7}}\exp\left(-v\frac{\kappa_8}{\lambda_{SN}}-\frac{\Gamma}{\lambda_{SN}\theta\gamma}\right)\frac{2}{\lambda_{SB} \lambda_{BN}} K_0\left(2\sqrt{\frac{v}{\lambda_{SB} \lambda_{BN}}}\right)}_{\mathcal{I}_L}dv=\mathcal{I}_F-\mathcal{I}_L,\nonumber			\tag{C.5}
	\end{align}
	\setcounter{equation}{\value{MYtempeqncnt}}
	\hrulefill 
\end{figure*} 
	\section{Proof of Theorem \ref{theorem_OPF}}\label{Proof_of_OPF}
	\setcounter{equation}{0}
Apparently, when ${a_F}/{a_N}\leq\gamma_{F,th}^{s_F}$,  $P_{\text{out},F}=1$, thus we turn to the derivation when ${a_F}/{a_N}>\gamma_{F,th}^{s_F}$. Then we have
\begin{align}\label{calculation_PoutF}
&P_{\text{out},F}=\Pr\left(\gamma_{U_F}^{s_F}\leq\gamma_{F,th}^{s_F}\right)\nonumber\\
&=\Pr\left(|h_{SF}|^2\kappa_2\leq|h_{SB}|^2|h_{BF}|^2\kappa_1+\gamma_{F,th}^{s_F}\right)\nonumber\\
&\overset{(a)}{=}\Pr\left(|h_{SF}|^2\leq \nu_F\frac{\kappa_1}{\kappa_2}+\frac{\gamma_{F,th}^{s_F}}{\kappa_2}\right)\nonumber\\
&=\int_{0}^{\infty}\left[1-\exp\left(-\frac{\kappa_1}{\kappa_2\lambda_{SF}}-\frac{\gamma_{F,th}^{s_F}}{\kappa_2\lambda_{SF} }\right)\right]f_{\nu_F}(v)dv\nonumber\\
&\overset{(b)}{=}\!\!1\!\!-\!\!\frac{2\!\exp\!\left(\!-\!\mathcal{B}_0\right)}{\lambda_{SB}\lambda_{BF}}\!\!\int_{0}^{\infty}\!\!\exp\left(\!\!-\frac{\kappa_1v}{\kappa_2\lambda_{SF}}\!\!\right)\!\!K_0\!\!\left(\!\sqrt{\frac{4    v}{\lambda_{SB}\lambda_{BF}}}\!\right)\!\!dv,
\end{align}
where $\mathcal{B}_0$ has been defined in Table \ref{OP_parameters},
	\begin{subequations}
	\begin{align}	
	\kappa_1&\triangleq\beta^2\gamma\gamma_{F,th}^{s_F}\left[\theta+\eta(1-\theta)\right], \label{kappa1}\\
	\kappa_2&\triangleq\theta\gamma\left(a_F-a_N\gamma_{F,th}^{s_F}\right),\label{kappa2}
	\end{align}
\end{subequations} 
and 
$(a)$	can be arrived by $|h_{SB}|^2|h_{BF}|^2\triangleq \nu_F$; (b) follows from the  fact that \cite{zhang2019backscatter}:
\begin{align}
f_{\nu_F}(v)=\frac{2}{\lambda_{SB}\lambda_{BF}}K_0\left(2\sqrt{\frac{v}{\lambda_{SB}\lambda_{BF}}}\right).
\end{align} Finally, by applying [\citen{gradshteyn2014table}, Eq. (6.614.4)], we
arrive at (\ref{exp_of_OPF}), completing the proof.

	\section{Proof of Theorem \ref{theorem_OPN}}\label{Proof_of_theorem_OPN}
\setcounter{equation}{0}
By substituting (\ref{define_of_gammaUN_sF}) and (\ref{define_of_gammaUN_sN}) into (\ref{case_when_OPN_occur}), we can easily arrive at
\begin{align}\label{first_step_of_PoutN}
&P_{\text{out},N}=1-\Pr\left(\gamma_{U_N}^{s_F}\geq\gamma_{N,th}^{s_F},\gamma_{U_N}^{s_N}\geq\gamma_{N,th}^{s_N}\right),\nonumber\\
&=1-\Pr\left(|h_{SN}|^2\kappa_4\geq |h_{SB}|^2|h_{BN}|^2\kappa_3+\gamma_{N,th}^{s_F},\right.\nonumber\\
&\quad\quad\quad\quad\left.|h_{SN}|^2\kappa_6\geq |h_{SB}|^2|h_{BN}|^2\kappa_5+\gamma_{N,th}^{s_N}\right),
\end{align}
where
\begin{subequations}
	\begin{align}	
		\kappa_3&\triangleq\beta^2\gamma\gamma_{N,th}^{s_F}\left[\theta+\eta(1-\theta)\right], \label{kappa3}\\
		\kappa_4&\triangleq\theta\gamma\left(a_F-a_N\gamma_{N,th}^{s_F}\right),\label{kappa4}\\
		\kappa_5&\triangleq\beta^2\gamma\gamma_{N,th}^{s_N}\left[\theta+\eta(1-\theta)\right], \label{kappa5}\\
		\kappa_6&\triangleq\theta\gamma a_N.\label{kappa6}
	\end{align}
\end{subequations} 
Clearly, when $\kappa_4\leq 0$, $ \gamma_{U_N}^{s_F} $ is always less than $\gamma_{N,th}^{s_F}$, leading to $P_{\text{out},N} = 1 
$. Therefore, the following derivations are based on the assumption that $\kappa_4\geq 0$, i.e., ${a_F}/{a_N}\geq\gamma_{N,th}^{s_F}$. For convenience, we define $|h_{SB}|^2|h_{BN}|^2\triangleq \nu_N$, then (\ref{first_step_of_PoutN}) can be further simplified as
\begin{align}\label{second_step_of_PoutN}
&P_{\text{out},N}=1-\Pr\left(|h_{SN}|^2\geq \nu_N\max\left(\frac{\kappa_3}{\kappa_4},\frac{\kappa_5}{\kappa_6}\right)\right.\nonumber\\
&\quad\quad\quad\quad\quad\quad\quad\left.+\max\left(\frac{\gamma_{N,th}^{s_F}}{\kappa_4},\frac{\gamma_{N,th}^{s_N}}{\kappa_6}\right)\right).\nonumber\\
&\overset{(a)}{=}\Pr\left(|h_{SN}|^2\leq \nu_N\frac{\beta^2\Gamma\left[\theta+\eta(1-\theta)\right]}{\theta}+\frac{\Gamma}{\theta\gamma}\right),
\end{align}
where step $(a)$ involves parameter substitution according to the following two equations:
\begin{subequations}
	\begin{align}	
		&\max\left(\frac{\kappa_3}{\kappa_4},\frac{\kappa_5}{\kappa_6}\right)\triangleq\frac{\beta^2\Gamma\left[\theta+\eta(1-\theta)\right]}{\theta},\label{maxk4_k5}\\
		&\max\left(\frac{\gamma_{N,th}^{s_F}}{\kappa_4},\frac{\gamma_{N,th}^{s_N}}{\kappa_6}\right)\triangleq\frac{\Gamma}{\theta\gamma}, \label{max_gamma_k4}
	\end{align}
\end{subequations} 
in which $\Gamma$ has been defined in Table \ref{OP_parameters}.
Considering that the PDF of $\nu_N$ can be expressed as:
\begin{align}\label{PDF_of_nuN}
	f_{\nu_N}(v) = \frac{2}{\lambda_{SB} \lambda_{BN}} K_0\left(2\sqrt{\frac{v}{\lambda_{SB} \lambda_{BN}}}\right),
\end{align}
and noting that (\ref{second_step_of_PoutN}) has similar form as (\ref{calculation_PoutF}), we can employ similar computational steps outlined in Appendix \ref{Proof_of_OPF} to evaluate  (\ref{second_step_of_PoutN}). In this way, we  finally obtain (\ref{exp_of_OPN}), completing the proof.

\section{Proof of Theorem \ref{theorem_OPC}}\label{Proof_of_theorem_OPC}
\setcounter{equation}{0}
By substituting (\ref{define_of_gammaUN_sF}), (\ref{define_of_gammaUN_sN}), and (\ref{define_gammaUNSC}) into (\ref{cases_when_sc_cannot_decode}), we have
\begin{align}\label{first_step_of_Pout_C}
	&P_{\text{out},C}=1\!-\!\Pr\left(\gamma_{U_N}^{s_F}\geq\gamma_{N,th}^{s_F},\gamma_{U_N}^{s_N}\geq\gamma_{N,th}^{s_N},\gamma_{U_N}^{s_C}\geq\gamma_{N,th}^{s_C}\right),\nonumber\\
	&=1-\Pr\left(|h_{SN}|^2\kappa_4\geq |h_{SB}|^2|h_{BN}|^2\kappa_3+\gamma_{N,th}^{s_F},\right.\nonumber\\
	&\quad\quad\quad\quad\left.|h_{SN}|^2\kappa_6\geq |h_{SB}|^2|h_{BN}|^2\kappa_5+\gamma_{N,th}^{s_N},\right.\nonumber\\
	&\quad\quad\quad\quad\left.|h_{SB}|^2|h_{BN}|^2\kappa_7\geq\gamma_{N,th}^{s_C}\right),
\end{align} where $\kappa_3$, $\kappa_4$, $\kappa_5$, and $\kappa_6$ have been defined in (\ref{kappa3}) to (\ref{kappa6}), and
	\begin{align}	
		\kappa_7&\triangleq\beta^2\gamma\left[\theta-\eta(1-\theta)\gamma_{N,th}^{s_C}\right]. \label{kappa7}
	\end{align}
 
 It is straightforward that when $\kappa_4$ or $\kappa_7\leq 0$, $P_{\text{out},C} = 1 
$. Therefore, the subsequent derivations are conducted under the condition that $\kappa_4$ and $\kappa_7>0$, i.e., ${a_F}/{a_N}>\gamma_{N,th}^{s_F}$ and  $\gamma_{N,th}^{s_C}<{\theta}/{(\eta(1-\theta))}$. Then, by defining $|h_{SB}|^2|h_{BN}|^2 \triangleq \nu_N$ and utilizing (\ref{maxk4_k5}) as well as (\ref{max_gamma_k4}), we can  further simplify (\ref{first_step_of_Pout_C}) as (\ref{first_simple_of_PoutC}), shown at the top of the previous page. Please note that the PDF of $\nu_N$ has already been given in (\ref{PDF_of_nuN}).

Next, we need to find the theoretical solution for $\mathcal{Q}$. Letting 
\begin{align}	\tag{C.4}	
	\kappa_8&\triangleq\frac{\beta^2\Gamma\left[\theta+\eta(1-\theta)\right]}{\theta},\nonumber \label{kappa8}
\end{align}  	
$\mathcal{Q}$ can be rewritten as (\ref{compose_of_Q}), shown at the top of the previous page. It is evident that $\mathcal{I}_F$ can be solved through the method shown in Appendix \ref{Proof_of_theorem_OPN}.
Then, $\mathcal{I}_F$ can be calculated as
	\begin{align}\label{IF}
			\tag{C.6}
		\mathcal{I}_F=-\mathcal{A}_1\exp\left(\mathcal{A}_1\!-\!\mathcal{B}_1\right){\rm Ei}\left(-\mathcal{A}_1\right),\nonumber
	\end{align}
where $\mathcal{A}_1$ and $\mathcal{B}_1$ have been defined in Table \ref{OP_parameters}.

Then, we turn to $\mathcal{I}_L$.
However,  $\mathcal{I}_L$  involves the product of $K_0(\cdot)$ and $\exp(\cdot)$, making it challenging to obtain a closed-form expression for the integral. To evaluate this integral,  we can employ  the Gauss-Chebyshev quadrature method \cite{gradshteyn2014table}, resulting in
\begin{align}\label{IL}
		\tag{C.7}
\mathcal{I}_L\approx\mathcal{A}_2\sum_{i=1}^{n}\sqrt{1-t_i^2}\exp(-\mathcal{B}_2)K_0\left(\mathcal{C}_2\right),\nonumber
\end{align}
where $n$ represents the calculation accuracy, and $\mathcal{A}_2$,  $\mathcal{B}_2$, $\mathcal{C}_2$, as well as $t_i$ have been defined in Table \ref{OP_parameters}. Finally, by substituting  (\ref{IF}) and  (\ref{IL}) into (\ref{first_simple_of_PoutC}), we can derive the approximated expressions of $P_{\text{out},C}$, completing the proof.

\section{Proof of Theorem \ref{theorem_IPF}}\label{Proof_of_theorem_IPF}
\setcounter{equation}{0}
From (\ref{IP_define}), we observe that calculating the IP requires the CDF of $\gamma_E^{s_i}$ ($i \in \{N,F,C\}$). However, since $\gamma_E^{s_i}$ is defined in terms of $\gamma_V^{s_i}$, we must first derive the CDF of $\gamma_V^{s_i}$. Following similar steps as outlined in Appendix~\ref{Proof_of_OPF}, we can derive the CDF of $\gamma_V^{s_F}$ as follows:
\begin{align}\label{derivation_of_FgammaVsF}
	 &F_{\gamma_{V}^{s_F}}(x)= \Pr(\gamma_{V}^{s_F} \leq x) \nonumber \\
	& \overset{(a)}{=} 1 - \exp\left(-\frac{x}{\kappa_{9}\lambda_{SV}}\right)\int_{0}^{+\infty}\exp\left(-\frac{xv\beta^2\gamma}{\kappa_{9}\lambda_{SV}}\right)f_{\nu_V}(v)dv \nonumber \\
	& \overset{(b)}{=} 1 + \underbrace{\mathcal{A}_{e0}(x)\exp(\mathcal{A}_{e0}(x))\text{Ei}(-\mathcal{A}_{e0}(x))}_{\mathcal{G}_{e0}(x)}\exp(-\mathcal{B}_{e0}(x)) \nonumber \\
	& = 1 + \mathcal{G}_{e0}(x)\exp(-\mathcal{B}_{e0}(x)),
\end{align}
where 
\begin{subequations}
	\begin{align}	
		&\mathcal{A}_{e0}(x)= \frac{\kappa_{9}\lambda_{SV}}{\lambda_{SB}\lambda_{BV}x\beta^2\gamma},\\
		&\mathcal{B}_{e0}(x)= \frac{x}{\kappa_{9}\lambda_{SV}}.
	\end{align}
\end{subequations} 

In (\ref{derivation_of_FgammaVsF}), (a) is obtained by defining $|h_{SB}|^2 |h_{BV}|^2 \triangleq \nu_V$ and $\kappa_{9} = [\theta a_F - \theta a_N x - (1-\theta)x] \gamma$, while (b) is derived from the PDF of $\nu_V$, given by 
\begin{align}
	f_{\nu_V}(v) = \frac{2}{\lambda_{SB} \lambda_{BV}} K_0\left(2\sqrt{\frac{v}{\lambda_{SB}\lambda_{BV}}}\right),
\end{align}
along with [\citen{gradshteyn2014table}, Eq. (6.614.4)]. Notably, when $\kappa_{9} \leq 0$, we have $F_{\gamma_{V}^{s_F}}(x) = 1$. Therefore, we focus on the case where $\kappa_{9} > 0$ in the following derivations. Recall that $d_{SV}\approx d_{BV}$, we have $\lambda_{SV}\approx\lambda_{BV}$, then 
\begin{align}
	\mathcal{A}_{e0}(x)\approx\frac{\kappa_{9}}{\lambda_{SB}x\beta^2}.
\end{align}

Subsequently, since $\gamma_{E}^{s_F}=\max_{V\in\Phi_{\text{E}}}\{\gamma_{V}^{s_F}\}$, we can attain its CDF  as 
\begin{align}
	&F_{\gamma_E^{s_F}}(x)
	=\mathbb{E}_{\Phi_{\text{E}}}\left\{\prod\limits_{v\in\Phi_{\text{E}}\atop d_{SV}\geq r_p}F_{\gamma_{V}^{s_F}}(x)\right\}\nonumber\\
	&=\exp\left(-\lambda_e\int_{S}\left[1-F_{\gamma_{V}^{s_F}}(x)\right]dd_{SV}\right)\nonumber\\
	&\overset{(a)}{=}\exp\left(\lambda_e\mathcal{G}_{e0}(x)2\pi\int_{r_p}^{\infty}\exp(-\mathcal{M}_{e0}(x)r^{\alpha})rdr\right),
\end{align}
where (a) results from converting Cartesian coordinates to polar ones. Here,  $\mathcal{M}_{e0}(x) = {x}/{\kappa_{9}}$.
Finally, by applying [\citen{gradshteyn2014table}, Eq. (3.381.9)] along with (\ref{IP_define}) and letting \(x = \gamma_{E,th}^{s_F}\), we arrive at (\ref{IP_UF}), completing the proof.

\section{Proof of Theorem \ref{IP_BD}}\label{Proof_of_theorem_IPC}
\setcounter{equation}{0}
First of all, we should derive the CDF of $\gamma_{V}^{s_C}$. By defining $|h_{SB}|^2 |h_{BV}|^2 \triangleq \nu_V$, we can achieve
\begin{align}
	 &F_{\gamma_{V}^{s_C}}(x)=1 - \Pr(\gamma_{V}^{s_C} > x) \nonumber\\
	&= 1 - \underbrace{\Pr\left(|h_{SV}|^2 < \frac{\nu_V \delta(x)}{x(1 - \theta)} - \frac{1}{(1 - \theta) \gamma}\right)}_{\mathcal{Q}_{V}} \nonumber\\
	&= 1 - \mathcal{Q}_{V},
\end{align}
where $\delta(x) = \beta^2 [\theta - (1 - \theta)x]$. 
It is worth noting that the following conditions must be satisfied: $\nu_V > \frac{x}{\gamma \delta(x)}$ and $\delta(x)>0$;
 otherwise, the above equation will always equal to one.

Then we turn to solve $\mathcal{Q}_V$. After several mathematical operations, $\mathcal{Q}_V$ can be reformulated as
	\begin{align}\label{compose_of_QV}
	&\mathcal{Q}_{V}=\int_{\frac{x}{\gamma\delta(x)}}^{+\infty}F_{|h_{SV}|^2}(\frac{v\delta(x)}{x(1-\theta)}-\frac{1}{(1-\theta)\gamma})f_{\nu}(v)dv\nonumber\\
	&=\underbrace{\int_{0}^{+\infty}(1\!-\!\exp(-\frac{v\delta(x)}{\lambda_{SV}x(1-\theta)}\!+\!\frac{1}{\lambda_{SV}(1-\theta)\gamma}))f_{\nu}(v)dv}_{\mathcal{I}_{VF}}\nonumber\\
	&-\underbrace{\int_{0}^{\frac{x}{\gamma\delta(x)}}(1\!-\!\exp(-\frac{v\delta(x)}{\lambda_{SV}x(1-\theta)}\!+\!\frac{1}{\lambda_{SV}(1-\theta)\gamma}))f_{\nu}(v)dv}_{\mathcal{I}_{VL}}\nonumber\\
	&=\mathcal{I}_{VF}-\mathcal{I}_{VL}.
\end{align}
It is evident that $\mathcal{I}_{VF}$ can be solved through the similar method shown in Appendix \ref{Proof_of_theorem_IPF}.
\begin{align}\label{exp_IVF}
	\mathcal{I}_{VF} & = 1 \!\!-\!\! \int_{0}^{+\infty}\!\! \exp\left(-\frac{v \delta(x)}{\lambda_{SV} x (1 \!-\! \theta)}\!\! +\!\! \frac{1}{\lambda_{SV} (1 \!- \!\theta) \gamma}\right)\!\! f_{\nu_V}(v) dv \nonumber\\
	&= 1 + \underbrace{\mathcal{A}_{e2}(x) \exp(\mathcal{A}_{e2}(x)) \text{Ei}(-\mathcal{A}_{e2}(x))}_{\mathcal{G}_{e2}(x)} \exp(\mathcal{B}_{e2}(x)) \nonumber\\
	&= 1 + \mathcal{G}_{e2}(x) \exp(\mathcal{B}_{e2}(x)),
\end{align}
where 
\begin{subequations}
	\begin{align}	
		&\mathcal{A}_{e2}(x) \approx \frac{x (1 - \theta)}{\lambda_{SB}\delta(x)},\\
		&\mathcal{B}_{e2}(x) = \frac{1}{\lambda_{SV} (1 - \theta) \gamma}.
	\end{align}
\end{subequations} 

As for $\mathcal{I}_{VL}$, obtaining a closed-form solution for its current expression is highly challenging. Therefore, the Gauss-Chebyshev quadrature method acts as an excellent solution to approximate $\mathcal{I}_{VL}$ as follows:
\begin{align}\label{exp_IVL}
	\mathcal{I}_{VL} \approx \mathcal{A}_{e3} \sum_{i=1}^{n} \sqrt{1 - t_i^2} \left(1 - \exp(-\mathcal{B}_{e3})\right) K_0(\mathcal{C}_{e3}),
\end{align} 
where $n$ represents the calculation accuracy, $t_i$ has been defined in Table \ref{IP_parameters}, and
\begin{subequations}
	\begin{align}	
		&\mathcal{A}_{e3}(x)=\frac{x\pi}{2\gamma\delta(x)\lambda_{SB}\lambda_{BV}},\\
		&	\mathcal{B}_{e3}(x)=\frac{x(t_i+1)\delta(x)}{\gamma\delta(x)2\lambda_{SV}x(1-\theta)}-\frac{1}{\lambda_{SV}(1-\theta)\gamma},\\
		&\mathcal{C}_{e3}(x)=2\sqrt{\frac{x(t_i+1)}{\gamma\delta(x)2\lambda_{SB}\lambda_{BV}}}.
	\end{align}
\end{subequations} 
Then by substituting (\ref{exp_IVF}) and  (\ref{exp_IVL}) into (\ref{compose_of_QV}), we can obtain the approximated expressions of $\mathcal{Q}_V$.

Similar to the derivation of $P_{\text{int},F}$ in Appendix \ref{Proof_of_theorem_IPF}, after setting $x=\gamma_{E,th}^{s_C}$,
we can write $P_{\text{int},C}$ as
\begin{align}
	P_{\text{int},C}=1-\exp\left(-2\pi\lambda_e\underbrace{\int_{r_p}^{+\infty}\mathcal{Q}_{V}(r)rdr}_{\mathcal{I}_{\mathcal{Q}}}\right).
\end{align}

To the best of the authors' knowledge, the integral $\mathcal{I}_{\mathcal{Q}}$ is intractable. However, by setting \( l = r - r_p \), we can tightly approximated $\mathcal{I}_{\mathcal{Q}}$ by using Gaussian-Laguerre quadrature technique [\citen{abramowitz1948handbook}, Eq. (25.4.45)]:
\begin{align}
	P_{\text{int},C}\!\!=\!\!1\!\!-\!\!\exp\left(\!-2\pi\lambda_e\sum_{i=1}^{N}w_i\mathcal{Q}_{V}(l_i+r_p)(l_i+r_p)\exp(l_i)\right),
\end{align}
where  $l_i$ is the $i$-th root of the Laguerre polynomial, $L_{N}(l)$, and $w_i$ has been defined in Table \ref{IP_parameters}.
In this way, we finally arrive at (\ref{exp_of_IPC}), completing the proof.

\end{appendices}

\bibliographystyle{IEEEtran}
\bibliography{AmBC_NOMA_random.bib}
\end{document}